\documentclass[runningheads, envcountsame, a4paper]{llncs}

\usepackage[utf8]{inputenc}
\usepackage{amsmath}		%
\usepackage{amssymb}		%
\usepackage{verbatim} 
\usepackage{adjustbox}
\usepackage{todonotes} 
\usepackage{tikz}		
\usetikzlibrary{automata}	%
\usetikzlibrary{decorations.pathmorphing,arrows,positioning}
\usetikzlibrary{calc}
\usepackage{tikzscale}
\usetikzlibrary{shapes.geometric}
\usetikzlibrary{arrows,automata}
\usetikzlibrary{matrix}
\usetikzlibrary{decorations.markings}
\usetikzlibrary{decorations.shapes}
\usepackage{stmaryrd}		%
\usepackage{microtype}
\usepackage{bbm}			%
\usepackage{braket} 		%
\usepackage{enumerate}
\usepackage{upgreek}
\usepackage{array}
\let\olddefinition\definition
\renewcommand{\definition}{\olddefinition\normalfont}	%
\DeclareMathOperator{\const}{const}
\DeclareMathOperator{\Lab}{Lab}
\DeclareMathOperator{\free}{free}
\DeclareMathOperator{\wt}{wt}

\DeclareMathOperator{\MSO}{MSO}

\newcommand{\call}{\mathrm{call}}
\newcommand{\ret}{\mathrm{ret}}
\newcommand{\intt}{\mathrm{int}}
\newcommand{\pcall}{\mathrm{pcall}}

\newcommand{\push}{\mathrm{push}}
\newcommand{\pop}{\mathrm{pop}}
\newcommand{\shift}{\mathrm{shift}}
\DeclareMathOperator{\acc}{acc}
\DeclareMathOperator{\Fin}{Fin}
\DeclareMathOperator{\Tree}{Tree}
\DeclareMathOperator{\Next}{Next}
\DeclareMathOperator{\Succ}{Succ}
\DeclareMathOperator{\myIf}{IF}
\DeclareMathOperator{\myThen}{THEN}
\DeclareMathOperator{\myElse}{ELSE}

\newcommand{\avv}{\curvearrowright }

\makeatletter
\newcommand{\thickhline}{%
	\noalign {\ifnum 0=`}\fi \hrule height 1pt
	\futurelet \reserved@a \@xhline
}
\newcolumntype{"}{@{\hskip\tabcolsep\vrule width 1pt\hskip\tabcolsep}}
\makeatother
\begin{document}
\title{Weighted Operator Precedence Languages}
	\author{Manfred Droste\inst{1} \and
		Stefan D\"uck\inst{1}\thanks{supported by Deutsche Forschungsgemeinschaft (DFG) Graduiertenkolleg 1763 (QuantLA).} \and
		Dino Mandrioli\inst{2} \and
		Matteo Pradella\inst{2,3}
		}
	\institute{Institute of Computer Science, Leipzig University, D-04109 Leipzig, Germany
		\email{\{droste,dueck\}@informatik.uni-leipzig.de}
			\and
			Dipartimento di Elettronica, Informazione e Bioingegneria (DEIB), Politecnico di Milano, Piazza Leonardo Da Vinci 32, 20133 Milano, Italy
		\email{\{dino.mandrioli,matteo.pradella\}@polimi.it}
        \and 
        IEIIT, Consiglio Nazionale delle Ricerche, via Ponzio 34/5, 20133 Milano, Italy
	}
\authorrunning{M. Droste, S. D\"uck, D. Mandrioli, and
	M. Pradella}
\toctitle{Weighted Automata and Logics for Operator Precedence Languages}
\tocauthor{Manfred Droste, Stefan D\"uck, Dino Mandrioli, and
	Matteo Pradella}
\maketitle
\begin{abstract}
In the last years renewed investigation of operator precedence languages (OPL) led to discover important properties thereof: OPL are closed with respect to all major operations, are characterized, besides the original grammar family, in terms of an automata family and an MSO logic; furthermore they significantly generalize the well-known visibly pushdown languages (VPL).
In another area of research, quantitative models of systems are also
greatly in demand. In this paper, we lay the foundation to marry these two research fields. We introduce weighted operator precedence automata and show how they are both strict extensions of OPA and weighted visibly pushdown automata.
We prove a Nivat-like result which shows that quantitative OPL can be described by unweighted OPA and very particular weighted OPA.
In a B\"uchi-like theorem, we show that weighted OPA are expressively equivalent to a weighted MSO-logic for OPL.
\end{abstract}
	\section{Introduction}
In the long history of formal languages the family of regular languages (RL), those that are recognized by finite state machines (FSM) or are generated by regular grammars, has always played a major role: thanks to its simplicity and naturalness it enjoys properties that are only partially extended to larger families. Among the many positive results that have been achieved for RL (e.g., expressiveness, decidability, minimization, ...), those of main interest in this paper are the following:
\begin{itemize}\renewcommand{\labelitemi}{$\bullet$}
	\item RLs have been characterized in terms of various mathematical logics. The pioneering papers are due to B\"uchi, Elgot, and Trakhtenbrot \cite{Bue60,Elg61,Tra61}  who independently developed a monadic second order (MSO) logic defining exactly the RL family. This work too has been followed by many further results; in particular those that exploited weaker but simpler logics such as first-order, propositional, and temporal ones which culminated in the breakthrough of model checking to support automatic verification \cite{McNaughtPap71,Emerson90,DBLP:journals/acta/ChoffrutMMP12}. 
	\item Weighted RLs have been introduced by Sch\"utzenberger in his pioneering paper \cite{Sch61}: by assigning a weight in a suitable algebra to each language word,
 we may specify several attributes of the word, e.g., relevance, probability, etc. Much research then followed and extended the original Sch\"utzenberger's work in various directions, cf. the books \cite{BR88,Eil74,KS86,SS78,DKV09}.
\end{itemize} 
Unfortunately, all families with greater expressive power than RL --typically context-free languages (CFL), which are the most widely used family in practical applications-- pay a price in terms of properties and, consequently, of possible tools supporting their automatic analysis. For instance, for CFL, the containment problem is undecidable and they are not closed under complement.

What was not possible for general CFL, however, has been possible for important subclasses of this family, which together we call \emph{structured CFL}. Informally, with this term we denote those CFLs where the syntactic tree-structure of their words is immediately ``visible'' in the words themselves. A first historical example of such families is that of parenthesis languages, introduced by McNaughton in another seminal paper \cite{McNaughton67}, which are generated by grammars whose right hand sides are enclosed within pairs of parentheses; not surprisingly an equivalent formalism of parenthesis grammars was soon defined, namely tree-automata which generalize the basics of FSM to tree-like structures instead of linear strings \cite{Tha67}. Among the many variations and generalizations of parenthesis languages the recent family of \emph{input-driven languages (IDL)} \cite{DBLP:conf/icalp/Mehlhorn80,Input-driven}, alias \emph{visibly pushdown languages (VPL)} \cite{AM09}, have received much attention in recent literature. For most of these structured CFL, including in particular IDL, all of the algebraic properties of RL still hold \cite{AM09}. One of the most noticeable results of this research field has been a characterization of IDL/VPL in terms of a MSO logic that is a fairly natural extension of the original B\"uchi's one for RL \cite{Lautemann94,AM09}.

This fact has suggested to extend the investigation of weighted RL to various cases of structured languages. The result of such a fertile approach is a rich collection of \emph{weighted logics}, first studied by Droste and Gastin \cite{DG07}, associated with \emph{weighted tree automata} \cite{DV06} and \emph{weighted VPAs} the automata recognizing VPLs, also called weighted NWAs \cite{Mat10,DD16}. %
\par
In an originally unrelated way \emph{operator precedence languages (OPL)} have been defined and studied in two phases temporally separated by four decades. In his seminal work \cite{Floyd1963} Floyd was inspired by the precedence of multiplicative operations over additive ones in the execution of arithmetic expressions and extended such a relation to the whole input alphabet in such a way that it could drive a deterministic parsing algorithm that builds the syntax tree of any word that reflects the word's semantics; Fig. \ref{fig:exp} and Section \ref{section:preliminaries} give an intuition of how an OP grammar generates arithmetic expressions and assigns them a natural structure. After a few further studies \cite{RMM78}, OPL's theoretical investigation has been abandoned due to the advent of LR grammars which, unlike OPL grammars, generate all deterministic CFL. %

OPL, however, enjoy a distinguishing property which we can intuitively describe as "\emph{OPL are input driven but not visible}". They can be claimed as \emph{input-driven} since the parsing actions on their words --whether to push or to pop their stack-- depend exclusively on the input alphabet and on the relation defined thereon, but their structure is \emph{not visible} in their words: e.g, they can include unparenthesized arithmetic expressions where the precedence of multiplicative operators over additive ones is explicit in the syntax trees but hidden in their frontiers (see Fig. \ref{fig:exp}). Furthermore, unlike other structured CFL, OPL include deterministic CFL that are not real-time \cite{LMPP15}. 

This remark suggested to resume their investigation systematically at the light of the recent technological advances and related challenges. Such a renewed investigation led to prove their closure under all major language operations \cite{Crespi-ReghizziM12} and to characterize them, besides the original Floyd's grammars, in terms of an appropriate class of pushdown automata (OPA) and in terms of a MSO logic which is a fairly natural but not trivial extension of the previous ones defined to characterize RL and VPL \cite{LMPP15}. Thus, OPL enjoy the same nice properties of RL and many structured CFL but considerably extend their applicability by breaking the barrier of visibility and real-time push-down recognition.

In this paper we put together the two above research fields, namely we introduce \emph{weighted OPL} and show that they are able to model system behaviors that cannot be specified by means of less powerful weighted formalisms such as weighted VPL. For instance, one might be interested in the behavior of a system which handles calls and returns but is subject to some emergency interrupts. Then it is important to evaluate how critically the occurrences of interrupts affect the normal system behavior, e.g., by counting the number of pending calls that have been preempted by an interrupt.
As another example consider a system logging all hierarchical calls and returns over words where this structural information is hidden. Depending on changing exterior factors like energy level, such a system could decide to log the above information in a selective way.
  \par Our main contributions in this paper are the following.
\begin{itemize}\renewcommand{\labelitemi}{$\bullet$}
\item The model of \emph{weighted OPA}, which have semiring weights at their transitions, significantly increases the descriptive power of previous weighted extensions of VPA, and has desired closure and robustness properties.
	\item For arbitrary semirings, there is a relevant difference in the expressive power of the model depending on whether it permits assigning weights to pop transitions or not.
	For commutative semirings, however, weights on pop transitions do not increase the expressive power of the automata. The difference in descriptive power between weighted OPA with arbitrary weights and without weights at pop transitions is due to the fact that OPL may be non-real-time and therefore OPA may execute several pop moves without advancing their reading heads.
	\item An extension of the classical result of Nivat \cite{Niv68} to weighted OPL. This robustness result shows that the behaviors of weighted OPA without weights at pop transitions are exactly those that can be constructed from weighted OPA with only one state, intersected with OPL, and applying projections which preserve the structural information. 
	\item A weighted MSO logic and, for arbitrary semirings, a Büchi-Elgot-Trakhtenbrot-Theorem proving its expressive equivalence to weighted OPA without weights at pop transitions.
	As a corollary, for commutative semirings this weighted logic is equivalent to weighted OPA including weights at pop transitions.
\end{itemize}
\par
	\section{Preliminaries} \label{section:preliminaries}

We start with an example to provide an intuition of the idea by which R. Floyd made the hidden precedences between symbols occurring in a grammar explicit in parse trees \cite{Floyd1963}: consider arithmetic expressions with two
operators, an additive one and a multiplicative one that takes precedence over the other one, in the sense that, during the interpretation of the expression, multiplications must be executed before sums. Parentheses are used to force different precedence hierarchies.  
Figure~\ref{fig:exp} (left) presents a grammar and (center) the derivation tree of the expression 
$n + n \times (n + n)$; all nonterminals are axioms.

Notice that the structure of the syntax tree (uniquely) corresponding to the input expression reflects
the precedence order which drives computing the value attributed to the expression. 
This structure, however, is not immediately visible in the expression;
if we used a parenthesis grammar, it would produce the string $(n+ (n \times ( n + n) ) )$ 
instead of the previous one, and the structure of the corresponding tree would be immediately visible. 
For this reason we say that such grammars ``hide'' the structure associated with a sentence, whereas parenthesis grammars and other input-driven ones make the structure explicit in the sentences they generate.

\begin{figure}
\begin{center}
\begin{tabular}{m{0.3\textwidth}m{0.3\textwidth}m{0.3\textwidth}}
\quad 
$\begin{array}{l}
E\to E + T \mid T\\ 
T\to T \times F \mid F\\
F\to n \mid ( E )
\end{array}$
&
\begin{tikzpicture}[scale=0.4]
        \node {$E$}
	child {
		node {$E$}
		child{ 
			node{$T$}
			child{ 
						node{$F$}
						child{ 
									node{$n$}
								}
					}
				}	
		}
	child {
		node {$+$}
	}
	child {
		node {$T$}
		child {
			node {$T$}
			child{ 
				node{$F$}
				child {
					node {$n$}
			                edge from parent %
			              }
					}
				}
		child{
			node {$\times$}
		}
		child{
			node {$F$}
			child{ 
						node{$($}
					}
			child{ 
						node{$E$}
						child{ 
									node{$E$}
									child{ 
												node{$T$}
												child{ 
															node{$F$}
															child{ 
																		node{$n$}
																	}
														}
									}			
								}
						child{ 
									node{$+$}
								}
						child{ 
									node{$T$}
									child{ 
												node{$F$}
												child{ 
															node{$n$}
														}	
											}	
								}						
					}
			child{ 
						node{$)$}
					}
				}	
		}
;
\end{tikzpicture}
&
$
\begin{array}{c|cccccc}
    &+ &\times & ( & ) & n  \\
\hline
+ & \gtrdot &\lessdot &\lessdot &\gtrdot &\lessdot \\
\times  & \gtrdot &\gtrdot &\lessdot &\gtrdot &\lessdot \\
( & \lessdot &\lessdot &\lessdot &\doteq &\lessdot \\
) & \gtrdot &\gtrdot & &\gtrdot \\
n & \gtrdot &\gtrdot & &\gtrdot \\
\end{array}
$
\end{tabular}
\end{center}
\caption{A grammar generating arithmetic expressions (left), an example derivation tree (center), and the precedence matrix (right).}
\label{fig:exp}
\end{figure}
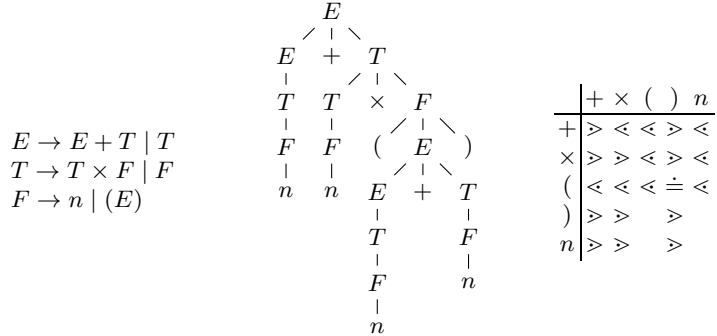

\par To model this hierarchical structure and make it accessible, we introduce the \emph{chain relation} $\curvearrowright$.
This new relation can be compared with the \emph{nesting} or \emph{matching relation} of \cite{AM09}, as it also is a non-crossing relation, going always forward and originating from additional information on the alphabet. However, it also features significant differences: Instead of adding unary information to symbols, which partition the alphabet into three disjoint parts (calls, internals, and returns), we add a binary relation for every pair of symbols denoting their precedence relation. Therefore, in contrast to the nesting relation, the same symbol can be either call or return depending on its context. %
Furthermore, the same position can be part of multiple chain relations.
\par
More precisely, we define an \emph{OP alphabet} as a pair $(\Sigma,M)$, where $\Sigma$ is an alphabet and $M$, the \emph{operator precedence matrix (OPM)} is a $|\Sigma\cup\{\#\}|^2$ array describing for each ordered pair of symbols at most one (operator precedence) relation, that is, every entry of $M$ is either $\lessdot$ (\emph{yields precedence}), $\doteq$ (\emph{equal in precedence}, $\gtrdot$ (\emph{takes precedence}), or empty (no relation).

\par
We use the symbol $\#$ to mark the beginning and the end of a word and let always be $\# \lessdot a$ and $a \gtrdot \#$ for all $a \in \Sigma$.
As an example, Figure~\ref{fig:exp} (right) depicts the OPM of the grammar reported on its left, omitting the standard relations for $\#$.

\par
Let $w=(a_1...a_n) \in \Sigma^+$ be a word. We say $a_0=a_{n+1}=\#$ and define a new relation $\curvearrowright$ on the set of all positions of $\#w\#$, inductively, as follows. Let $i,j \in \{0,1,...,n+1\}$, $i < j$.
Then, we write $i \curvearrowright j$
	if there exists a sequence of positions ${k_1}...{k_m}$ such that $i=k_1<...<k_m=j$,
	$a_{k_1} \lessdot a_{k_2} \doteq ... \doteq a_{k_{m-1}} \gtrdot a_{k_m}$, and either $k_s + 1 = k_{s+1}$ or $k_s \curvearrowright k_{s+1}$
	for each $s \in \{1,...,m-1\}$.
In particular, $i \curvearrowright j$ holds if $a_i \lessdot a_{i+1} \doteq ... \doteq a_{j-1} \gtrdot a_j$.
\par
We say $w$ is \emph{compatible} with $M$ if for $\#w\#$ we have $0 \curvearrowright n+1$.
In particular, this forces $M_{{a_i}{a_j}} \neq \emptyset$ for all $i+1=j$ and for all $i \curvearrowright j$.
We denote by $(\Sigma,M)^+$ the set of all non-empty
words
over $\Sigma$ which are compatible with $M$.
For a \emph{complete} OPM $M$, i.e. one without empty entries, this is $\Sigma^+$.
\par
We recall the definition of an operator precedence automaton from \cite{LMPP15}.
\begin{definition} \label{def:OPA}
	A \emph{(nondeterministic) operator precedence automaton (OPA)} $\mathcal{A}$ over an OP alphabet $(\Sigma, M)$ %
    is a tuple $\mathcal{A}=(Q,I,F,\delta)$, where $\delta = (\delta_{\shift},\delta_{\push},\delta_{\pop})$, consisting of 
	\begin{itemize}
		\item Q, a finite set of states, 
		\item $I \subseteq Q$, the set of initial states,
		\item $F \subseteq Q$, a set of final states, and 
		\item the transition relations $\delta_{\shift},\delta_{\push} \subseteq Q\times \Sigma \times Q$, and $\delta_{\pop} \subseteq Q\times Q \times Q$.
	\end{itemize}
\end{definition}
	Let $\Gamma=\Sigma\times Q$. A \emph{configuration} of $\mathcal{A}$ is a triple $C= \langle \Uppi, q, w\# \rangle$, where $\Uppi \in \bot \Gamma^*$ represents a stack, $q\in Q$ the current state, and $w$ the remaining input to read.
	\par
	A \emph{run} of $\mathcal{A}$ on $w=a_1...a_n$ is a finite sequence of configurations $C_0 \vdash ... \vdash C_m$ such that every transition $C_i \vdash C_{i+1}$ has one of the following forms, where %
    $a$ is the topmost alphabet symbol of $\Uppi$ and $b$ is the next symbol of the input to read:
    \begin{center}$\begin{array}{rrll}
      \textit{push move}:&
          \langle \Uppi, q, bx\rangle
          &\vdash
          \langle \Uppi[b,q], r, x\rangle
          &\text{ if } a\lessdot b \text{ and } (q,b,r)\in \delta_{\push},\\
      \textit{shift move}:&
          \langle \Uppi[a,p], q, bx\rangle
          &\vdash
          \langle \Uppi[b,p], r, x\rangle
          &\text{ if } a\doteq b \text{ and } (q,b,r)\in \delta_{\shift},\\
      \textit{pop move}:&
          \langle \Uppi[a,p], q, bx\rangle
          &\vdash
          \langle \Uppi, r, bx\rangle
          &\text{ if } a\gtrdot b \text{ and } (q,p,r)\in \delta_{\pop} .
    \end{array}$\end{center}
	An \emph{accepting run} of $\mathcal{A}$ on $w$ is a run from $\langle \bot, q_I, w\#\rangle$ to $\langle \bot, q_F, \#\rangle$, where $q_I \in I$ and $q_F \in F$.
     The \emph{language accepted by $\mathcal{A}$}, denoted $L(\mathcal{A})$, consists of all words over $(\Sigma,M)^+$ which have an accepting run on $\mathcal{A}$.
	We say that $L\subseteq(\Sigma,M)^+$ is an \emph{OPL} if $L$ is accepted by an OPA over $(\Sigma,M)$.
	As proven by \cite{LMPP15}, the deterministic variant of an OPA, using a single initial state instead of $I$ and transition functions instead of relations, is equally expressive to nondeterministic OPA.

   An example automaton is depicted in Figure~\ref{fig:exprAut}: with the OPM of Figure~\ref{fig:exp} (right), it accepts the same language as the grammar of Figure~\ref{fig:exp} (left).
\begin{figure}
\begin{center}
\begin{tikzpicture}[every edge/.style={draw,solid}, node distance=4cm, auto, 
                    every state/.style={draw=black!100,scale=0.6}, >=latex] %

\node[initial by arrow, initial text=,state] (q0) {{\huge $0$}};
\node[state] (q1) [right of=q0, xshift=0cm, accepting] {{\huge $1$}};
\node[state] (q2) [below of=q0, xshift=0cm] {{\huge $2$}};
\node[state] (q3) [right of=q2, xshift=0cm, accepting] {{\huge $3$}};

\path[->]
(q0) edge [below] node {$n$} (q1)
(q0) edge [bend right, right]  node {$($} (q2)
(q1) edge [loop above, double, right]  node {$\ 0, 1$} (q1)
(q1) edge [bend right, above]  node {$+, \times$} (q0)

(q2) edge [below ] node {$n$} (q3)
(q2) edge [loop left, left] node {$($} (q2)
(q3) edge [loop above, right, double]  node {$\ 0, 1, 2, 3$} (q3) 
(q3) edge [bend right, above ]  node {$+, \times$} (q2)
(q3) edge [loop right, right, dashed] node {$)$} (q3);
\end{tikzpicture}
\caption{Automaton for the language of the grammar of Figure~\ref{fig:exp}. Shift, push and pop transitions are denoted by dashed, normal and double arrows, respectively.}\label{fig:exprAut}
\end{center}
\end{figure}
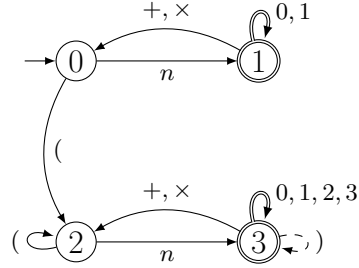

\begin{definition}%
	The logic $\MSO(\Sigma,M)$, short $\MSO$, is defined as
	\begin{align*}
		\beta&::=\Lab_a(x)~|~x\leq y~|~x\curvearrowright y~|~x \in X~|~\neg \beta~|~\beta \vee \beta~|	~\exists x.\beta~|~\exists X.\beta
	\end{align*}
	where $a\in \Sigma \cup \{\#\}$, $x,y$ are first-order variables; and $X$ is a second order variable.
\end{definition}
We define the natural semantics for this (unweighted) logic as in \cite{LMPP15}. The relation $\curvearrowright$ refers to the chain relation %
introduced above.
\begin{theorem}[\cite{LMPP15}]
	\label{thm:OPLMSO}
	A language $L$ over $(\Sigma,M)$ is an OPL iff it is MSO-definable.
\end{theorem}
\section{Weighted OPA and Their Connection to Weighted VPA}
\label{section:weightedOPA}
In this section, we introduce a weighted extension of operator precedence automata. We show that weighted OPL include weighted VPL and give examples showing how these weighted automata can express behaviors which were not expressible before.
Let $\mathbb{K}=(K,+,\cdot,0,1)$ be a \emph{semiring}, i.e., $(K,+,0)$ is a commutative monoid, $(K,\cdot,1)$ is a monoid, $(x+y)\cdot z = x\cdot z + y\cdot z$, $x\cdot (y+z) = x\cdot y + x\cdot z$, and $0 \cdot x = x \cdot 0 = 0$ for all $x,y,z \in K$. $\mathbb{K}$ is called \emph{commutative} if $(K,\cdot,1)$ is commutative. \par
Important examples of commutative semirings cover the Boolean semiring $\mathbb{B}=(\{0,1\},\vee,\wedge,0,1)$, the semiring of the natural numbers $\mathbb{N}=(\mathbb{N},+,\cdot,0,1)$, or the tropical semirings $\mathbb{R}_{\max}=(\mathbb{R} \cup \{-\infty\}, \max, +, -\infty, 0)$ and $\mathbb{R}_{\min}=(\mathbb{R} \cup \{\infty\}, \min, +, \infty, 0)$. Non-commutative semirings are given by $n\times n$-matrices over semirings $\mathbb{K}$ with matrix addition and multiplication as usual ($n \geq 2$), or the semiring $(\mathcal{P}(\Sigma^*),\cup, \cdot, \emptyset, \{\varepsilon\})$ of languages over $\Sigma$.
\begin{definition}
	A \emph{weighted OPA (wOPA)} $\mathcal{A}$ over an OP alphabet $(\Sigma,M)$ and a semiring $\mathbb{K}$ is a tuple
	$\mathcal{A}=(Q,I,F,\delta,\wt)$, where $\wt=(\wt_{\shift},\wt_{\push},\wt_{\pop})$, consisting of
	\begin{itemize}
		\item an OPA $\mathcal{A}'=(Q,I,F,\delta)$ over $(\Sigma,M)$ and 
		\item the weight functions $\wt_{op}:\delta_{op} \rightarrow K$, $op \in \{\shift,\push,\pop\}$.
	\end{itemize}
	We call a wOPA \emph{restricted}, denoted by \emph{rwOPA}, if $\wt_{\pop} \equiv 1$, i.e. $\wt_\pop(q,p,r) = 1$ for each $(q,p,r)\in \delta_\pop$.
\end{definition}	
	A \emph{configuration} of a wOPA is a tuple $C= \langle \Uppi, q, w\#, k \rangle$, 
    where $(\Uppi,q,w\#)$ is a configuration of the OPA $\mathcal{A}'$ %
    and $k\in \mathbb{K}$.
	A \emph{run} of $\mathcal{A}$ is a again a sequence of configurations $C_0 \vdash C_1 \ldots \vdash C_m$ satisfying the previous conditions and, additionally, the weight of a configuration is updated by multiplying with the weight of the encountered transition, as follows. As before, we denote with $a$ %
    the topmost symbol of $\Uppi$ and with $b$ the next symbol of the input to read:
    \begin{center}$\begin{array}{rll}
          \langle \Uppi, q, bx, k\rangle
          &\vdash
          \langle \Uppi[b,q], r, x, k \cdot \wt_{\push}(q,b,r)\rangle
          &\text{ if } a\lessdot b \text{ and } (q,b,r)\in \delta_{\push},\\
          \langle \Uppi[a,p], q, bx, k\rangle
          &\vdash
          \langle \Uppi[b,p], r, x, k \cdot \wt_{\shift}(q,b,r) \rangle
          &\text{ if } a\doteq b \text{ and } (q,b,r)\in \delta_{\shift},\\
          \langle \Uppi[a,p], q, bx, k\rangle
          &\vdash
          \langle \Uppi, r, bx, k \cdot \wt_{\pop}(q,p,r) \rangle
          &\text{ if } a\gtrdot b \text{ and } (q,p,r)\in \delta_{\pop} .
    \end{array}$\end{center}
We call a run $\rho$ \emph{accepting} if it goes from $\langle \bot, q_I, 1, w\#\rangle$ to $\langle \bot, q_F, k, \#\rangle$, where $q_I \in I$ and $q_F \in F$.
For such an accepting run, the \emph{weight of} $\rho$ is defined as $\wt(\rho)=k$.
We denote by $\acc(\mathcal{A},w)$ the set of all accepting runs of $\mathcal{A}$ on $w$.
\par
Finally, the \emph{behavior of} $\mathcal{A}$ is a function $\llbracket\mathcal{A}\rrbracket:(\Sigma,M)^+ \rightarrow K$, defined as
\begin{align*}
	\llbracket A \rrbracket(w)=\sum_{\rho \in \acc(\mathcal{A},w)}\wt(\rho)\enspace.
\end{align*}
Every function $S:(\Sigma,M)^+ \rightarrow K$ is called an \emph{OP-series} (short: \emph{series}, also \emph{weighted language}).
A wOPA $\mathcal{A}$ \emph{recognizes} or \emph{accepts} a series $S$ if $\llbracket\mathcal{A}\rrbracket = S$.
A series $S$ is called \emph{regular} or a \emph{wOPL} if there exists an wOPA $\mathcal{A}$ accepting it. $S$ is \emph{strictly regular} or an \emph{rwOPL} if there exists an rwOPA $\mathcal{A}$ accepting it.
%\todo{S: added wOPL, rwOPL, cf short version. Here I think we have the freedom to use different notions}

\begin{example} \label{example:penalty}
	Let us resume, in a simplified version, an example presented in \cite{LMPP15} (Example 8) which exploits the ability of OPA to pop many items from the stack without advancing the input head: in this way we can model a system that manages calls and returns in a traditional LIFO policy but discards all pending calls if an interrupt occurs\footnote{A similar motivation inspired the recent extension of VPL as colored nested words by \cite{AF16}.}.
    The weighted automaton of Figure \ref{figure:wOPA1} attaches weights to the OPA's transitions in such a way that the final weight of a string is 1 only if no pending call is discarded by any interrupt; otherwise, the more calls are discarded the lower the ``quality'' of the input as measured by its weight.
	\par
	More precisely, we define $\Sigma=\{\call,\ret,\intt\}$ and the precedence matrix $M$ as a subset of the matrix of Example 8 of \cite{LMPP15}, i.e., $\call\lessdot\call$, $\call\doteq\ret$, $\call\gtrdot\intt$, $\intt\lessdot\intt$, $\intt \gtrdot \call$, and $\ret \gtrdot a$ for all $a\in \Sigma$.
	\par
	By adopting the same graphical notation as in \cite{LMPP15} pushes are normal arrows, shifts are dashed, pops are double arrows; weights are given in brackets at transitions.
	\begin{figure}
		\begin{center}
		\begin{tikzpicture}[->,>=latex, auto, node distance = 3.5cm]
			\node [state, initial left, initial text =, accepting ] (A) {$q_0$};
			\path
			(A) edge [out=145,in=115,looseness=8] node [above] {$\call(\frac{1}{2})$~} (A)
			(A) edge [dashed, out=65,in=35,looseness=8] node [above] {~$\ret(2)$} (A)
			(A) edge [double, out=285,in=255,looseness=8] node [below] {$q_0(1)$} (A)
			(A) edge [out=0,in=330,looseness=8] node {$\intt(1)$} (A)
			;
		\end{tikzpicture}
		\end{center}
		\caption{The weighted OPA $\mathcal{A}_{\text{penalty}}$ penalizing unmatched calls}
		\label{figure:wOPA1}
	\end{figure}
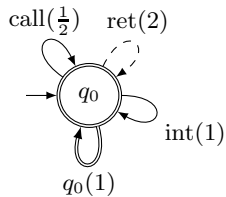
	Let $\#\pcall(w)$ be the number of pending calls of $w$, i.e., calls which are never answered by a return. Then the behavior of the automaton $\mathcal{A}_{\text{penalty}}$ over	$(\Sigma,M)$ and the semiring $(\mathbb{N},+,\cdot,0,1)$ given in Figure \ref{figure:wOPA1} is
	$\llbracket\mathcal{A}_{\text{penalty}}\rrbracket(w)
			=(\frac{1}{2})^{\#\pcall(w)}$.
	\par
	The example can be easily enriched by following the same path outlined in \cite{LMPP15}: we could add symbols specifying the serving of an interrupt, add different types of calls and interrupts with different priorities and more sophisticated policies (e.g., lower level interrupts disable new calls but do not discard them, whereas higher level interrupts reset the whole system, etc.)
\end{example}
\begin{example} \label{example:policy}
	The wOPA of Figure \ref{figure:wOPA1} is ``rooted'' in a deterministic OPA; thus the semiring of weights is exploited in a fairly trivial way since only the $\cdot$ operation is used.
    The automaton $\mathcal{A}_{\text{policy}}$ given in Figure \ref{figure:wOPA2}, instead, formalizes a more complex system where the penalties for unmatched calls may change nondeterministically within intervals delimited by the special symbol $\$$. Precisely, the symbols $\$$ mark intervals during which sequences of calls, returns, and interrupts occur; ``normally'' unmatched calls are not penalized, but there is a special, nondeterministically chosen interval during which they are penalized; the global weight assigned to an input sequence is the maximum over all nondeterministic runs that are possible when recognizing the sequence.
	\par
	Here, the alphabet %
	is $\Sigma=\{\call,\ret,\intt,\$\}$, and the OPM $M$, with $a \lessdot \$$ and $\$ \gtrdot a$, for all $a \in \Sigma$ is a natural extension of %
    the OPM of Example \ref{example:penalty}. As semiring, we take $\mathbb{R}_{\max}=(\mathbb{R} \cup \{-\infty\}, \max, +, -\infty, 0)$.
	Then, $\llbracket\mathcal{A}_{\text{policy}}\rrbracket(w)$ equals the maximal number of pending calls between two consecutive $\$$.
	\begin{figure}
	\begin{center}
		\begin{tikzpicture}[->,>=latex, auto, node distance = 3.5cm]
		\node [state, initial left, initial text = ] (X) {$q_0$};
		\node [state, right of = X ] (A) {$q_1$};
		\node [state, right of = A, accepting ] (Y) {$q_2$};
		\path
		(X) edge [out=145,in=115,looseness=8] node [above] {$\$(0),\intt(0),\call(0)$\qquad\qquad} (X)
		(X) edge [dashed, out=65,in=35,looseness=8] node [right] {~$\ret(0)$} (X)
		(X) edge [double,out=285,in=255,looseness=8] node [below] {$q_0(0)$} (X)
		(X) edge node [below] {$\$(0)$} (A)
		(A) edge [out=165,in=135,looseness=8] node [above] {$\call(1)$~} (A)
		(A) edge [dashed, out=105,in=75,looseness=8] node [above] {~$\ret(-1)$} (A)
		(A) edge [out=45,in=15,looseness=8] node [above] {$~~\intt(0)$} (A)
		(A) edge [double,out=285,in=255,looseness=8] node [below] {$q_0(0),q_1(0)$} (A)
		(A) edge node [below] {$\$(0)$} (Y)
		(Y) edge [out=145,in=115,looseness=8] node [above] {$\$(0),\call(0)$\quad} (Y)
		(Y) edge [dashed, out=65,in=35,looseness=8] node [right] {~$\ret(0)$} (Y)
		(Y) edge [loop right] node {$\intt(0)$} (Y)
		(Y) edge [double,out=285,in=255,looseness=8] node [below] {$q_0(0),q_1(0),q_2(0)$} (Y)
		;
		\end{tikzpicture}
	\end{center}
	\caption{The weighted OPA $\mathcal{A}_{\text{policy}}$ penalizing unmatched calls nondeterministically} \label{figure:wOPA2}
	\end{figure}
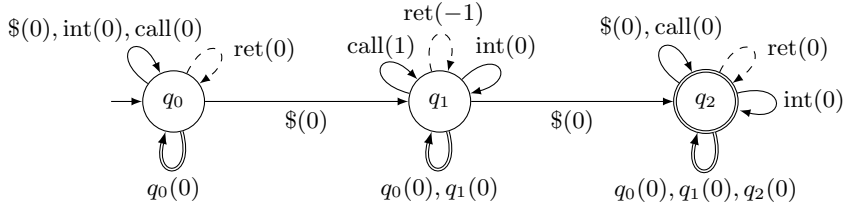
	Again, $\mathcal{A}_{\text{policy}}$ can be easily modified/enriched to formalize several variations of its policy: e.g., different policies could be associated with different intervals, different weights could be assigned to different types of calls and/or interrupts, different policies could also be defined by choosing different semirings, etc.
\end{example}
Note that both automata, %
$\mathcal{A}_{\text{penalty}}$ and $\mathcal{A}_{\text{policy}}$, do not use the weight assignment for pops.
\begin{example} \label{example:log}
	The next automaton $\mathcal{A}_{\text{log}}$, depicted in Figure \ref{figure:wOPA3} chooses non-deter\-ministi\-cally between logging everything and logging only `important' information, e.g., only interrupts (this could be a system dependent on energy, WiFi, ...). Notice that, unlike the previous examples, in this case assigning nontrivial weights to $\pop$ transitions is crucial.
	\par
    Let $\Sigma=\{\call,\ret,\intt\}$, and define $M$ as for $\mathcal{A}_{\text{penalty}}$.
	We employ the semiring $(\Fin_{\Sigma'},\cup,\circ,\emptyset,\{\varepsilon\})$
	of all finite languages over $\Sigma'=\{c,r,p,i\}$.
	Then, $\llbracket\mathcal{A}_{\text{log}}\rrbracket(w)$ yields all possible logs on $w$.	
	\begin{figure} 
		\begin{center}
		\begin{tikzpicture}[->,>=latex, auto, node distance = 4.0cm
			]
			\node [state, initial left, initial text =, accepting ] (X) {$q_0$};
			\node [state, right of = X, accepting] (A) {$q_1$};
			\path
			(X) edge [out=145,in=115,looseness=8] node
				[above,align=center] %
				{$\call(c)$ \\
				int$(i)$} (X)
			(X) edge [dashed, out=75,in=45,looseness=8] node [above] {$\ret(r)$} (X)
			(X) edge [double,out=285,in=255,looseness=8] node [below] {$q_0(p)$} (X)
			(X) edge node [above] {$\call(\varepsilon)$%
				}(A)
			(A) edge [bend left] node [below] {$\call(\varepsilon)$%
				}(X)
			(A) edge [out=145,in=115,looseness=8] node
				[above,align=center] %
				{$\call(\varepsilon)$ \\
				int$(i)$} (A)
			(A) edge [dashed, out=75,in=45,looseness=8] node [above] {$\ret(\varepsilon)$} (A)
			(A) edge [double,out=285,in=255,looseness=8] node [below] {$q_0(\varepsilon),q_1(\varepsilon)$} (A)
			;
		\end{tikzpicture}
		\end{center}
		\caption{The wOPA  $\mathcal{A}_{\text{log}}$ nondeterministically writes logs at different levels of detail.} \label{figure:wOPA3}
	\end{figure}
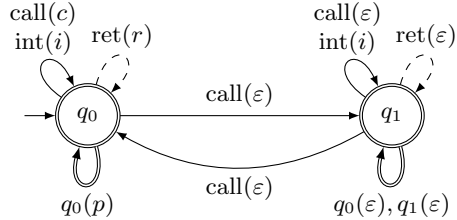	
\end{example}
As hinted at by our last example, the following proposition shows that in general, wOPA are more expressive than rwOPA.
\begin{proposition} There exists an OP alphabet $(\Sigma,M)$ and a semiring $\mathbb{K}$ such that there exists a weighted language $S$ %
which is regular but not strictly regular.
\end{proposition}
\begin{proof}
	Let $\Sigma=\{c,r\}$, $c\lessdot c$, and $c\doteq r$.
	Consider the semiring $\Fin_{\{a,b\}}$ of all finite languages over $\{a,b\}$ together with union and concatenation.
	Let $n\in \mathbb{N}$ and $S:(\Sigma,M)^+\rightarrow \Fin_{\{a,b\}}$ be the following series
	\begin{align*}
		S(w) = \left \{ \begin{array}{ll}
					\{a^nba^n\} &,\text{ if } w=c^nr \\
					\emptyset &, \text{ otherwise}
				\end{array} \right . \enspace .
	\end{align*}
	Then, we can define a wOPA which only reads $c^nr$, assigns the weight $\{a\}$ to every push and pop, and the weight $\{b\}$ to the one shift, and therefore accepts $S$, as in Figure \ref{figure:wOPA4}.
	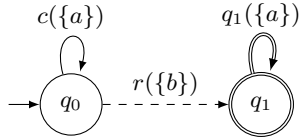
\begin{figure} 
		\begin{center}
		\begin{tikzpicture}[->,>=latex, auto, node distance = 2.5cm
			]
			\node [state, initial left, initial text = ] (X) {$q_0$};
			\node [state, right of = X, accepting] (A) {$q_1$};
			\path
			(X) edge [loop above, right] node[above] {$c(\{a\})$} (X)
			(X) edge [dashed] node [above] {$r(\{b\})$}(A)
			(A) edge [loop above, right, double] node [above] {$q_1(\{a\})$} (A)
			;
		\end{tikzpicture}
		\end{center}
		\caption{The wOPA recognizing $S(c^nr)=\{a^nba^n\}$ and $S(w)=0$, otherwise.} \label{figure:wOPA4}
	\end{figure}
	\par
	Now, we show with a pumping argument that there exists no rwOPA which recognizes $S$. Assume there is an rwOPA $\mathcal{A}$ with $\llbracket\mathcal{A}\rrbracket = S$. Note that for all $n\in \mathbb{N}$, the structure of $c^nr$ is fixed as $c \lessdot c \lessdot ... \lessdot c \doteq r$.
	Let $\rho$ be an accepting run of $\mathcal{A}$ on $c^nr$ with $\wt(\rho)=\{a^nba^n\}$.
	Then, the transitions of $\rho$ consist of $n$ pushes, followed by a shift, followed by $n$ pops and can be written as
	\[q_0 \stackrel{c}{\longrightarrow} q_1
	\stackrel{c}{\longrightarrow} ...
	\stackrel{c}{\longrightarrow} q_{n-1}
	\stackrel{c}{\longrightarrow} q_{n}
	\stackrel{r}{\dashrightarrow} q_{n+1} \stackrel{q_{n-1}}{\Longrightarrow} q_{n+2}
	\stackrel{q_{n-2}}{\Longrightarrow} ...
	\stackrel{q_{1}}{\Longrightarrow} q_{2n}
	\stackrel{q_{0}}{\Longrightarrow} q_{2n+1} \enspace.\]
	Both the number of states and the amount of pairs of states are bound. If $n$ is sufficiently large, there exists two pop transitions $\pop(q,p,r)$ and $\pop(q',p',r')$ in this sequence such that $q=q'$ and $p=p'$. This means that we have a loop in the pop transitions going from state $q$ to $q'=q$. Furthermore, the corresponding push to the first transition of this loop was invoked when the automaton was in state $p'$, while the corresponding push to the last pop was invoked in state $p$. Since $p=p'$, we also have a loop at the corresponding pushes.
	Then, the run where we skip both loops in the pops and in the pushes is an accepting run for $c^{n-k}r$, for some $k\in\mathbb{N}\setminus\{0\}$.
	\par
	Since the weight of all pops is trivial, the weight of the pop-loop is ${\varepsilon}$. If the weight of the push-loop is also ${\varepsilon}$, then we have an accepting run for $c^{n-k}r$ of weight $\{a^nba^n\}$, a contradiction. If the weight of the push-loop is not trivial, then by a simple case distinction it has to be either $\{a^i\}$ for some $i \in \mathbb{N}\setminus\{0\}$ or it has to contain the $b$. In the first case, the run without both loops has weight $\{a^{n-i}ba^n\}$ or $\{a^{n}ba^{n-i}\}$, in the second case it has weight $\{a^j\}$, for some $j \in \mathbb{N}$.  All these runs are not of the form $a^{n-k}ba^{n-k}$, a contradiction.
\qed
\end{proof}
	We notice that using the same arguments, we can show that also no weighted nested word automata as defined in \cite{Mat10,DPi14} can recognize this series.
	Even stronger, we can prove that restricted weighted OPLs are a
	generalization of weighted VPLs in the following sense.
	We shortly recall the important definitions.
	 Let $\Sigma=\Sigma_{\call} \sqcup \Sigma_{\intt} \sqcup \Sigma_{\ret}$ be a \emph{visibly pushdown alphabet}. A \emph{VPA} is a pushdown automata which uses a push and pop transitions whenever it reads a call or return symbol, respectively.
	 \par
	 In \cite{Crespi-ReghizziM12}, it was shown that using the complete OPM of Fig. \ref{figure:VPLMatrix}, for every VPA, there exists an equivalent operator precedence grammar which in turn can be transformed into an equivalent OPA.
	\begin{figure}
	\begin{center}	
		\bgroup
		\def\arraystretch{1.5}
		\begin{tabular}{c " c | c | c |}
			 & $~\Sigma_{\call}~$ & $~\Sigma_{\ret}~$ & $~\Sigma_{\intt}~$ \\
			\thickhline
			$~\Sigma_{\call}~$ & $\lessdot$ & $\doteq$ & $\lessdot$\\
			\hline
			$\Sigma_{\ret}$  & $\gtrdot$ & $\gtrdot$ & $\gtrdot$\\
			\hline
			$\Sigma_{\intt}$ & $\gtrdot$ & $\gtrdot$ & $\gtrdot$\\
			\hline
		\end{tabular}
		\egroup
	\end{center}
	\caption{OPM for VPL} \label{figure:VPLMatrix}
	\end{figure}
	\par
	In \cite{Mat10} and \cite{DPi14} weighted extensions of VPA were introduced (in the form of \emph{weighted nested word automata wNWA}). These add semiring weights at every transition again depending on the information what symbols are calls, internals, or returns. Note that every nested word has a representation as a word over a visibly pushdown alphabet $\Sigma$ and therefore can be seen as a compatible word of $(\Sigma,M)^+$, where $M$ is the OPM of Fig. \ref{figure:VPLMatrix}, i.e., we can interpret the behavior of a wNWA as an OP-series $(\Sigma,M)^+ \rightarrow \mathbb{K}$.

\begin{theorem} \label{thm:wNWArwOPA}
	Let $\mathbb{K}$ be a semiring, $\Sigma$ be a visibly pushdown alphabet, and $M$ be the OPM of Fig.~\ref{figure:VPLMatrix}.
	Then for every wNWA $\mathcal{A}$ defined as in \cite{DPi14}, there exists an rwOPA $\mathcal{B}$ %
	with $\llbracket\mathcal{A}\rrbracket(w) = \llbracket\mathcal{B}\rrbracket(w)$ for all $w\in(\Sigma,M)^+$.
\end{theorem}
	We give an intuition for this result as follows.
	Note that although sharing some similarities, pushes, shifts, and pops are not the same thing as calls, internals, and returns. Indeed, a return of a (w)NWA reads and 'consumes' a symbol, while a pop of an (rw)OPA just pops the stack and leaves the next symbol untouched. %
	\par
	After studying Figure \ref{figure:VPLMatrix}, this leads to the important observation that every symbol of $\Sigma_{\ret}$ and therefore every return transition of an NWA is simulated not by a pop, but by a shift transition of an OPA (in the unweighted and weighted case).
	\par
	We give a short demonstrating example:
	Let $\Sigma_{\intt}=\{a\}$, $\Sigma_{\call}=\{\langle c\}$, $\Sigma_{\ret}=\{r\rangle\}$, $w= a\langle car\rangle$.
	Then every run of an NWA for this word looks like
	\[q_0 \xrightarrow{~~~\,~a~~~~} q_1 \xrightarrow{\,\langle c\,} q_2 \xrightarrow{~~~~~a~~~~} q_3 \xrightarrow{~~~\,~r\rangle~~~~} q_4 \enspace.\]
	Every run of an OPA (using the OPM of Fig. \ref{figure:VPLMatrix}) looks as follows:
	\[q_0 \stackrel{a}{\longrightarrow} q_1' \stackrel{}{\Rightarrow} q_1 \stackrel{\langle c}{\longrightarrow} q_2 \stackrel{a}{\longrightarrow} q_3' \stackrel{}{\Rightarrow} q_3 \stackrel{r\rangle}{\dashrightarrow}  q_4' \stackrel{}{\Rightarrow}q_4\enspace,\]
	where the return was substituted (by the OPM, not by a choice of ours) by a shift followed by a pop.
	\par
	It follows that we can simulate a weighted call by a weighted push, a weighted internal by a weighted push together with a pop and a weighted return by a weighted shift together with a pop. Therefore, we may indeed omit weights at pop transitions.
\begin{proof}[of Theorem \ref{thm:wNWArwOPA}]
	Given a weighted NWA $\mathcal{A}=(Q,I,F,(\delta_{\call},\delta_{\intt},\delta_{\ret}),$ $(\wt_\call, \wt_\intt, \wt_\ret))$ over $\Sigma$ and $\mathbb{K}$,
	we construct an rwOPA $\mathcal{B}=(Q',I',F',$ $(\delta_{\push},\delta_{\shift},\delta_{\pop}), (\wt_\push',\wt_\shift',\wt_\pop'))$ over $(\Sigma,M)$ and $\mathbb{K}$. We set $Q'=Q\cup (Q\times Q)$, $I'=I$, and $F'=F$.
    We define the relations $\delta_{\push}$, $\delta_{\shift}$, $\delta_{\pop}$, and the functions $\wt_\push'$, $\wt_\shift'$, and $\wt_\pop'$ as follows.
    \par
    We let $\delta_\push$ contain all triples $(q,a,r)$ with $(q,a,r)\in\delta_{\call}$, and all triples $(q,a,(q,r))$ with $(q,a,r)\in\delta_{\intt}$. We set $\wt_\push'(q,a,r) = \wt_\call(q,a,r)$ and $\wt_\push'(q,a,(q,r)) = \wt_\intt(q,a,r)$.
    Moreover, we let $\delta_\shift$ contain all triples $(q,a,(p,r))$ with $(q,p,a,r)\in\delta_{\ret}$ and set $\wt_\shift'(q,a,(p,r)) = \wt_\ret(q,p,a,r)$.
    Furthermore, we let $\delta_\pop$ contain all triples $((q,r),q,r)$ with $(q,a,r)\in\delta_{\intt}$, and all triples $((p,r),p,r)$ with $(q,p,a,r)\in\delta_{\ret}$, and set $\wt_\pop'((q,r),q,r)=\wt_\pop'((p,r),p,r)=1$.
	\par
	Then, a run analysis of $\mathcal{A}$ and $\mathcal{B}$ shows that  $\llbracket\mathcal{B}\rrbracket = \llbracket\mathcal{A}\rrbracket$.
	\qed
\end{proof}

Together with the result that OPA are strictly more expressive than VPAs \cite{Crespi-ReghizziM12},
this gives a complete picture of the expressive power of these three classes of weighted languages:
\[
	\textrm{wVPL} \subsetneq \textrm{rwOPL} \subsetneq \textrm{wOPL} \enspace.
\]
The following result shows that for commutative semirings the second part of this hierarchy collapses, i.e. restricted rwOPA are equally expressive as wOPA (and therefore can be seen as a kind of normal form in this case).

\begin{theorem}
	Let $\mathbb{K}$ be a commutative semiring and $(\Sigma,M)$ an OP alphabet. Let $\mathcal{A}$ be a wOPA. Then, there exists an rwOPA $\mathcal{B}$ with $\llbracket\mathcal{A}\rrbracket=\llbracket\mathcal{B}\rrbracket$.
\end{theorem}
\begin{proof}
Let $\mathcal{A}=(Q,I,F,\delta,\wt)$ be a wOPA over $(\Sigma,M)$ and $\mathbb{K}$. Note that for every pop transition of a wOPA, there exists exactly one push transition. We construct an rwOPA $\mathcal{B}$ over the state set $Q'=Q\times Q \times Q$ and with the same behavior as $\mathcal{A}$ with the following idea in mind. In the first state component $\mathcal{B}$ simulates $\mathcal{A}$. In the second and third state component of $Q'$ the automaton $\mathcal{B}$ preemptively guesses the states $q$ and $r$ of the pop transition $(q,p,r)$ of $\mathcal{A}$ which corresponds to the next push transition following after this configuration. This enables us to transfer the weight from the pop transition to the correct push transition.
	\par
	The detailed construction of $\mathcal{B}=(Q',I',F',\delta',\wt')$ over $(\Sigma,M)$ and $\mathbb{K}$ is the following. If $Q=\emptyset$, then $\llbracket\mathcal{A}\rrbracket\equiv0$ is trivially strictly regular. If $Q$ is nonempty, let $q\in Q$ be a fixed state. Then, we set $Q'=Q\times Q \times Q$, $I'=\{(q_1,q_2,q_3) \mid q_1 \in I,q_2,q_3 \in Q\}$, $F'=\{(q_1,q,q) \mid q_1 \in F\}$, and
	\begin{align*}
		\delta'_{\push}&=\{((q_1,q_2,q_3),a,(r_1,r_2,r_3)) \mid (q_1,a,r_1) \in \delta_{\push} \text{ and } (q_2,q_1,q_3) \in \delta_{\pop}\} \\
		\delta'_{\shift}&=\{((q_1,q_2,q_3),a,(r_1,q_2,q_3)) \mid (q_1,a,r_1) \in \delta_{\shift}\} \\
        \delta'_{\pop}&=\{((q_1,q_2,q_3),(p_1,q_1,r_1),(r_1,q_2,q_3)) \mid (q_1,p_1,r_1) \in \delta_{\pop} \} \enspace .
	\end{align*}
    Here, every push of $\mathcal{B}$ controls that the previously guessed $q_2$ and $q_3$ can be used by a pop transition of $\mathcal{A}$ going from $q_2$ to $q_3$ with $q_1$ on top of the stack. %
    Every pop controls that the symbols on top of the stack are exactly the ones used at this pop. Since the second and third state component are guessed for the next push, they are passed on whenever we read a shift or pop. The second and third component pushed at the first position of a word are guessed by an initial state. At the last push, which therefore has no following push and will propagate the second and third component to the end of the run, the automaton $\mathcal{B}$ has to guess the distinguished state used in the final states.
	\par
    Therefore, %
	$\mathcal{B}$ has exactly one accepting run (of the same length) for every accepting run of $\mathcal{A}$, and vice versa. 
    Finally, we define the transition weights as follows.
	\begin{align*}
		\wt'_{\push}((q_1,q_2,q_3),a,(r_1,r_2,r_3))&=\wt_{\push}(q_1,a,r_1) \cdot \wt_{\pop}(q_2,q_1,q_3) \\
		\wt'_{\shift}((q_1,q_2,q_3),a,(r_1,r_2,r_3))&=\wt_{\shift}(q_1,a,r_1)\\
		\wt'_{\pop}& \equiv 1 \enspace.
	\end{align*}	
	Then, %
	the runs of $\mathcal{A}$ simulated by $\mathcal{B}$ have exactly %
    the same weights but in a different ordering. %
	Since $\mathbb{K}$ is commutative, it follows that $\llbracket\mathcal{A}\rrbracket=\llbracket\mathcal{B}\rrbracket$.
	\qed
\end{proof}
In the following, we study closure properties of weighted OPA and restricted weighted OPA. As usual, we extend the operation $+$ and $\cdot$ to series $S,T:(\Sigma,M)^+ \rightarrow K$ by means of pointwise definitions as follows:
\begin{align*}
(S + T)(w) &= S(w) + T(w) \mbox{ for each } w \in (\Sigma,M)^+\\
(S \odot T)(w) &= S(w) \cdot T(w) \mbox{ for each } w \in (\Sigma,M)^+ \enspace .
\end{align*}
\begin{proposition}
	\label{prop:disjsums}
	The sum of two regular (resp. strictly regular) series over $(\Sigma,M)^+$ is again regular (resp. strictly regular).
\end{proposition}
\begin{proof}
	We use a standard disjoint union of two (r)wOPA accepting the given series to obtain a (r)wOPA for the sum as follows.
	\par
	Let $\mathcal{A}=(Q,I,F,\delta,\wt)$ and $\mathcal{B}=(Q',I',F',\delta',\wt')$ be two wOPA over $(\Sigma,M)$ and $\mathbb{K}$.
	We construct a wOPA $\mathcal{C}=(Q'',I'',F'',\delta'',\wt'')$ over $(\Sigma,M)$ and $\mathbb K$ by defining $Q''=Q \sqcup Q'$, $I''=I\cup I$, $F''=F \cup F'$, $\delta''=\delta \cup \delta'$. The weight function is defined by
	\[
		\wt''(t)=\left \{
		\begin{array}{ll}
			\wt(t) &,\text{ if } t \in \delta \\
			\wt'(t) &, \text{ if } t \in \delta' \\
		\end{array} \right .
	\enspace .\]
	Then, $\llbracket \mathcal{C} \rrbracket = \llbracket \mathcal{A} \rrbracket + \llbracket \mathcal{B} \rrbracket$. Furthermore, if $\mathcal{A}$ and $\mathcal{B}$ are restricted, i.e. $\wt\equiv 1$ and $\wt'\equiv 1$, it follow that $\wt''\equiv 1$, and therefore $\mathcal{C}$ is also restricted.
	\qed
\end{proof}
\begin{proposition} \label{prop:restrictedprod}
	Let $S:(\Sigma,M)^+ \rightarrow K$ be a
regular (resp. strictly regular) series and $L\subseteq (\Sigma,M)^+$ an OPL.
	Then, the series
	$(S \cap L)(w)=
		\left \{
			\begin{array}{ll}
				S(w) &,\text{ if } w \in L \\
				0 &,\text{ otherwise}
			\end{array}
		\right \}
	$
	is regular (resp. strictly regular).
	Furthermore, if $\mathbb{K}$ is commutative, then the product of two regular (resp. strictly regular) series over $(\Sigma,M)^+$ is again regular (resp. strictly regular).
\end{proposition}

\begin{proof}
	We use a product construction of automata. 
	\par
	Let $\mathcal{A}=(Q,I,F,\delta,\wt)$ be a wOPA over $(\Sigma,M)$ and $\mathbb K$ with $\llbracket \mathcal{A} \rrbracket = S$ and let $\mathcal{B}=(Q',q'_0,F',\delta')$ be a deterministic OPA over $(\Sigma,M)$ with $L(\mathcal{B})=L$.
	We construct a wOPA $\mathcal{C}=(Q'',I'',F'',\delta'',\wt'')$ over $(\Sigma,M)$ and $\mathbb{K}$, with
	$\llbracket\mathcal{C}\rrbracket=(S \cap L)(w)=
		\left \{
			\begin{array}{ll}
			S(w) &,\text{ if } w \in L \\
			0 &,\text{ otherwise}
			\end{array}
		\right \}
	$, as follows.
	We define $Q''=Q \times Q'$, $I''=I \times \{q'_0\}$, $F''=F \times F'$, and %
	\begin{align*}
		\delta''_{\push}&=\{((q,q'),a,(r,r')) \mid (q,a,r) \in \delta_{\push} \text{ and } \delta'_{\push}(q',a)=r'\}\enspace ,\\
		\delta''_{\shift}&=\{((q,q'),a,(r,r')) \mid (q,a,r) \in \delta_{\shift} \text{ and } \delta'_{\shift}(q',a)=r'\}\enspace ,\\
		\delta''_{\pop}&=\{((q,q'),(p,p'),(r,r')) \mid (q,p,r) \in \delta_{\pop} \text{ and } \delta'_{\pop}(q',p')=r'\} \enspace .
	\end{align*}
	Then the weights of $\mathcal{C}$ are defined as
	\begin{align*}
		\wt''_{\push}((q,q'),a,(r,r'))= \wt_{\push}(q,a,r) \enspace ,\\
		\wt''_{\shift}((q,q'),a,(r,r'))= \wt_{\shift}(q,a,r) \enspace ,\\
		\wt''_{\pop}((q,q'),(p,p'),(r,r'))= \wt_{\pop}(q,p,r) \enspace .
	\end{align*}
	Note that %
    given a word $w$, the automata $\mathcal{A}$, $\mathcal{B}$, and $\mathcal{C}$ have to use pushes, shifts, and pops at the same positions. Hence, every accepting run of $\mathcal{C}$ on $w$ defines exactly one accepting run of $\mathcal{B}$ and exactly one accepting run of $\mathcal{A}$ on $w$ with matching weights, and vice versa.
	We obtain
	\begin{align*}
		\llbracket \mathcal{C} \rrbracket (w)
		&=\sum_{\rho \in \acc(\mathcal{C},w)}\wt(\rho)
		\\
		&=\sum_{\substack{
				\rho, \text{ such that } \\
				\rho_{\restriction Q} \in \acc(\mathcal{A},w)  \\
				\rho_{\restriction Q'} \in \acc(\mathcal{B},w)}}
		\wt(\rho) \\
		&= \left \{ \begin{array}{ll}
		\sum_{\rho \in \acc(\mathcal{A},w)}\wt(\rho)
		&,\text{ if the run of } \mathcal{B} \text{ on $w$ is accepting }  \\
		0 &, \text{ otherwise}
		\end{array} \right .\\
		&= (S \cap L) (w) \enspace .
	\end{align*}
	It follows that, $\llbracket \mathcal{C} \rrbracket=S \cap L$.
	\par
	For the second part of the proposition, %
	let $\mathcal{A}=(Q,I,F,\delta,\wt)$ and $\mathcal{B}=(Q',I',F',\delta',\wt')$ be two wOPA.
	We construct a wOPA $\mathcal{P}$ as $\mathcal{P}=(Q \times Q', I \times I', F \times F', \delta^\mathcal{P},\wt^\mathcal{P})$ where $\delta^\mathcal{P}=(\delta^\mathcal{P}_{\push}, \delta^\mathcal{P}_{\shift}, \delta^\mathcal{P}_{\pop})$ and set
\begin{align*}
		\delta^\mathcal{P}_{\push}&=\{((q,q'),a,(r,r')) \mid (q,a,r) \in \delta_{\push} \text{ and } (q',a,r') \in \delta'_{\push}\}\enspace ,\\
		\delta^\mathcal{P}_{\shift}&=\{((q,q'),a,(r,r')) \mid (q,a,r) \in \delta_{\shift} \text{ and } (q',a,r') \in \delta'_{\shift}\}\enspace ,\\
		\delta^\mathcal{P}_{\pop}&=\{((q,q'),(p,p'),(r,r')) \mid (q,p,r) \in \delta_{\pop} \text{ and } (q',p',r') \in  \delta'_{\pop}\} \enspace ,
\end{align*}
and
\begin{align*}
	\wt^\mathcal{P}_{\push}((q,q'),a,(r,r')) &= \wt'_{\push}(q,a,r) \cdot \wt''_{\push}(q',a,r')\enspace, \\
	\wt^\mathcal{P}_{\shift}((q,q'),a,(r,r')) &= \wt'_{\shift}(q,a,r) \cdot \wt''_{\shift}(q',a,r')\enspace, \\
	\wt^\mathcal{P}_{\pop}((q,q'),(p,p'),(r,r')) &= \wt'_{\pop}(q,p,r) \cdot \wt''_{\pop}(q',p',r')\enspace.
\end{align*}
It follows that $\llbracket \mathcal{P} \rrbracket=\llbracket \mathcal{A} \rrbracket \odot \llbracket \mathcal{B} \rrbracket$. Furthermore, if $\mathcal{A}$ and $\mathcal{B}$ are restricted, then so is $\mathcal{P}$.
	\qed
\end{proof}
Next, we show that regular series are closed under projections which preserve the OPM. For two OP alphabets $(\Sigma,M)$, $(\Gamma,M')$ and a mapping $h:\Sigma \rightarrow \Gamma$, %
we write $h:(\Sigma,M) \rightarrow (\Gamma,M')$ and say $h$ is \emph{OPM-preserving} if for all $\odot \in \{\lessdot, \doteq, \gtrdot\}$, we have $a \odot b$ if and only if $h(a)\odot h(b)$. %
We can extend such an $h$ to a function $h:(\Sigma,M)^+ \rightarrow (\Gamma,M')^+$ as follows. Given a word $w=(a_1a_2...a_n)\in (\Sigma,M)^+$, we define $h(w)%
=h(a_1a_2...a_n)
=h(a_1)h(a_2)...h(a_n)$.
Let $S:(\Sigma,M)^+ \rightarrow K$ be a series. Then, we define $h(S):(\Gamma,M')^+ \rightarrow K$ for each $\mathit{v}\in (\Gamma,M')^+$ by
\begin{align}
	\label{formula:hom} h(S)(\mathit{v})&=
		\sum_{\substack{ ~w \in (\Sigma,M)^+ \\ h(w)=\mathit{v} }}
			S(w) \enspace.
\end{align}
\begin{proposition} \label{prop:closurehom}
	Let $\mathbb{K}$ be a semiring, $S:(\Sigma,M)^+ \rightarrow K$ regular (resp. strictly regular), and $h:\Sigma \rightarrow \Gamma$ an OPM-preserving projection. Then, $h(S):(\Gamma,M')^+ \rightarrow K$ is regular (resp. strictly regular). %
\end{proposition}
\begin{proof}
	We follow an idea of \cite{DV12} and its application in \cite{DPi14} and \cite{DD16}. 
	Let $\mathcal{A}=(Q,I,F,\delta,\wt)$ %
	be a wOPA over $(\Sigma,M)$ and $\mathbb{K}$ with $\llbracket \mathcal{A} \rrbracket = S$.
	The main idea is to remember the last symbol read in the next state to distinguish different runs of $\mathcal{A}$ which would otherwise coincide in $\mathcal{B}$.
	We construct the wOPA $\mathcal{B}=(Q',I',F',\delta',\wt')$ over $(\Sigma,M)$ and $\mathbb{K}$ as follows.
	We set $Q'=Q \times \Sigma$, $I'=I \times \{a_0\}$ for some fixed $a_0 \in \Sigma$, and $F'=F\times \Sigma$.
	We define the transition relations $\delta'=(\delta'_{\push}, \delta'_{\shift}, \delta'_{\pop})$ for every $b \in \Gamma$ and $(q,a), (q',a'), (q'',a'') \in Q'$, as
	\begin{align*}	
	\delta'_{\push}&=\{((q,a),b,(q',a')) \mid (q,a',q') \in \delta_{\push} \textrm{ and } b=h(a') \} \enspace , \\
	\delta'_{\shift}&=\{((q,a),b,(q',a')) \mid (q,a',q') \in \delta_{\shift} \textrm{ and } b=h(a') \} \enspace , \\
	\delta'_{\pop}&=\{((q,a),(q',a'),(q'',a)) \mid (q,q',q'') \in \delta_{\pop} \} \enspace .	
	\end{align*}
	Then, the weight functions are defined by
	\begin{align*}	
	\wt'_{\push}((q,a),h(a'),(q',a')) &= \wt_{\push} (q,a',q')\enspace,\\
	\wt'_{\shift}((q,a),h(a'),(q',a')) &= \wt_{\shift} (q,a',q')\enspace, \\
	\wt'_{\pop}((q,a),(q',a'),(q'',a'')) &= \wt_{\pop} (q,q',q'') \enspace.
	\end{align*}
	Analogously to \cite{DPi14} and \cite{DD16}, this implies that for every run $\rho$ of $\mathcal{A}$ on $w$, there exists exactly one run $\rho'$ of $\mathcal{B}$ on $v$ with $h(w)=v$ and $\wt(\rho)$=$\wt(\rho')$.
	One difference to previous works is that a pop of a wOPA is not consuming the symbol. Therefore, we have to make sure to not change the symbol, which we are currently remembering while processing a pop.
	\par
	It follows that
	$\llbracket \mathcal{A}'\rrbracket (v) = h(\llbracket \mathcal{A} \rrbracket)(v)$, so $h(S)= \llbracket \mathcal{A}' \rrbracket $ is regular. 
	Furthermore, if $\mathcal{A}$ is restricted, then so is $\mathcal{B}$.
	\qed
\end{proof}

	\section{A Nivat Theorem} \label{section:nivat}
In this section, we establish a connection between weighted
OPLs and strictly regular series. We show that strictly regular series are exactly those series which can be derived from a restricted weighted OPA with only one state, intersected with an unweighted OPL, and using an OPM-preserving projection of the alphabet. 
\par
Let $h:\Sigma' \rightarrow \Sigma$ be a map between two alphabets. %
Given an OP alphabet $(\Sigma,M)$, we define $h^{-1}(M)$ by setting $h^{-1}(M)_{a'b'} = M_{h(a')h(b')}$ for all $a',b' \in \Sigma'$.
As $h$ is OPM-preserving, for every series $S:(\Sigma,M)^+ \rightarrow K$, we get a series $h(S):(\Sigma',h^{-1}(M))^+ \rightarrow K$, using the sum over all pre-images as in formula (\ref{formula:hom}).
\par
Let $\mathcal{N}(\Sigma,M,\mathbb{K})$ comprise all series $S:(\Sigma,M)^+ \rightarrow K$ for which there exist an alphabet $\Sigma'$, %
a map $h:\Sigma' \rightarrow \Sigma$, and a one-state rwOPA $\mathcal{B}$ over $(\Sigma',h^{-1}(M))$ and $\mathbb{K}$ and an OPL $L$ %
over $(\Sigma',h^{-1}(M))$ such that $S=h(\llbracket\mathcal{B}\rrbracket \cap L)$.
\par
Now, we show that every strictly regular series can be decomposed into the above introduced fragments. %
\begin{proposition} \label{prop:RECtoN}
	Let $S:(\Sigma,M)^+\rightarrow K$ be a series. If $S$ is strictly regular, then $S$ is in
	$\mathcal{N}(A,B,\mathbb{K})$.
\end{proposition}
\begin{proof}
	We follow some ideas of \cite{DK15} and \cite{DP14}.
	\par
	Let $\mathcal{A}=(Q,I,F,\delta,\wt)$ be a rwOPA over $(\Sigma,M)$ and $\mathbb{K}$ with $ \llbracket\mathcal{A}\rrbracket = S$.
	We set $\Sigma' = Q \times \Sigma \times Q$ %
	as the extended alphabet. The intuition is that $\Sigma'$ consists of the push and the shift transitions of $\mathcal{A}$. %
	Let $h$ be the projection of $\Sigma'$ to $\Sigma$ and let $M'=h^{-1}(M)$. %
	\par
	Let $L\subseteq (\Sigma',M')^+$ be the language consisting of all words $w'$ over the extended alphabet such that $h(w')$ has an accepting run on $\mathcal{A}$ which uses at every position the push, resp. the shift transition defined by the symbol of $\Sigma'$ at this position. %
	\par
	We construct the unweighted OPA $\mathcal{A}'=(Q',I',F',\delta')$ over $(\Sigma',M')$, accepting $L$, as follows. We set $Q'=Q$, $I'=I$, $F'=F$, and define $\delta'$ as follows
	\begin{align*}
		\delta'_\push &= \set{(q,(q,a,p),p) \mid (q,a,p) \in \delta_\push}\enspace,\\
		\delta'_\shift &= \set{(q,(q,a,p),p) \mid (q,a,p) \in \delta_\shift} \enspace, \\ %
		\delta'_\pop &= \delta_\pop \enspace .
	\end{align*}
	Hence, $\mathcal{A'}$ has an accepting run on a word $w'\in (\Sigma',M')^+$ if and only if $\mathcal{A}$ has an accepting run on $h(w')$, using the push and shift transitions defined by $w'$.
	\par
	We construct the one-state rwOPA $\mathcal{B}=(Q'',I'',F'',\delta'',\wt'')$ %
	over $(\Sigma',M')$ and $\mathbb{K}$ as follows. Set $Q''=I''=F''=\{q\}$, $\delta''_\push=\delta''_\shift=\{(q,a',q) \mid a' \in \Sigma'\}$, $\delta''_\pop=\{(q,q,q)\}$, $\wt''_\push(q,a',q)=\wt_\push(a')$, $\wt''_\shift(q,a',q)=\wt_\shift(a')$, %
	for all $a' \in \Sigma'$, and $\wt''_\pop(q,q,q)=1$.
	\par
	Let $\rho$ be a run of $w=a_1...a_n \in (\Sigma,M)^+$ on $\mathcal{A}$ and $\rho'$ a run of $w'=a'_1...a'_n \in  (\Sigma',M')^+$ on $\mathcal{B}$. We denote with $\wt_{\mathcal{A}}(\rho,w,i)$, resp. $\wt_{\mathcal{B}}(\rho',w',i)$, the weight of the push or shift transition used by the run $\rho$, resp. $\rho'$, at position $i$. Since $\mathcal{A}$ and $\mathcal{B}$ are restricted, for all their runs $\rho$, $\rho'$, we have
	$\wt(\rho) = \prod^{|w|}_{i=1}\wt_{\mathcal{A}}(\rho,w,i)$, resp. $\wt(\rho') = \prod^{|w'|}_{i=1}\wt_{\mathcal{B}}(\rho',w',i)$. 
	Furthermore, following its definition, the rwOPA $\mathcal{B}$ has exactly one run $\rho$ for every word $w' \in (\Sigma',M')$ and for all $h(w')=w$ and for all $i \in \{1...n\}$, we have $\wt_{\mathcal{B}}(\rho',w',i)$ = $\wt_{\mathcal{A}}(\rho,w,i)$.
	It follows that
	\begin{align*}
		h(\llbracket\mathcal{B}\rrbracket \cap L)(w)
		&=\sum_{\substack{w' \in (\Sigma',M')^+ \\ h(w')=w }}
			(\llbracket\mathcal{B}\rrbracket \cap L)(w') \\
		&=\sum_{\substack{w' \in L(\mathcal{A'}) \\ h(w')=w }} %
			\llbracket\mathcal{B}\rrbracket(w')\\
		&=\sum_{\rho \in \acc(\mathcal{A},w)}
			\prod^{|w|}_{i=1}\wt_{\mathcal{A}}(\rho,w,i) \\
		&=\sum_{\rho \in \acc(\mathcal{A},w)}
			\wt(\rho) \\
		&= \llbracket\mathcal{A}\rrbracket (w) = S(w) \enspace . %
	\end{align*} 
	Hence, $S = h(\llbracket\mathcal{B}\rrbracket\cap L)$, thus $S \in \mathcal{N}(\Sigma,M,\mathbb{K})$.
\qed
\end{proof}
Using this proposition and closure properties of series, we get the following Nivat-Theorem for weighted operator precedence automata.
\begin{theorem}
	Let $\mathbb{K}$ be a semiring and $S:(\Sigma,M)^+\rightarrow K$ be a series.
	Then $S$ is strictly regular %
	if and only if $S\in \mathcal{N}(\Sigma,M,\mathbb{K})$.
\end{theorem}
\begin{proof}
	The ``only if''-part of is immediate by Proposition \ref{prop:RECtoN}.
	\par
	For the converse, let %
	$\Sigma'$ be an alphabet, %
	$h:\Sigma' \rightarrow \Sigma$, $L\subseteq (\Sigma',h^{-1}(M))^+$ be an OPL,
	$\mathcal{B}$ a one-state rwOPA,	
	and $S=h(\llbracket\mathcal{B}\rrbracket \cap L)$.
	Then Proposition \ref{prop:restrictedprod} shows that
		$\llbracket\mathcal{B}\rrbracket \cap L$
	is strictly regular. Now, Proposition \ref{prop:closurehom} yields the result.
	\qed 
\end{proof}

	\section{Weighted MSO-Logic for OPL} \label{section:wMSO}

We use modified ideas from Droste and Gastin \cite{DG07}, also incorporating the distinction into an unweighted (boolean) and a weighted part by Bollig and Gastin~\cite{BG09}.

\begin{definition}
	We define the weighted logic $\MSO(\mathbb{K},(\Sigma,M))$, short  $\MSO(\mathbb{K})$, as
	\begin{align*}
	\beta&::=\Lab_a(x)~|~x\leq y~|~x\curvearrowright y~|~x \in X~|~\neg \beta~|~\beta \vee \beta~|	~\exists x.\beta~|~\exists X.\beta \\
	\varphi&::=\beta~|~
	k~|~\varphi \oplus \varphi~|~\varphi \otimes \varphi~|~%
	\textstyle\bigoplus_x\varphi~|~\bigoplus_X \varphi~|~\prod_x \varphi
	\end{align*}
	where $k \in \mathbb{K}$; %
	$x,y$ are first-order variables; and $X$ is a second order variable.
\end{definition}
We call $\beta$ boolean and $\varphi$ weighted formulas.
Let $w \in (\Sigma,M)^+$ and $\varphi \in \MSO(\mathbb{K})$. Following classical approaches for logics , we denote with $[w] = \{1,...,|w|\}$ the set of all positions of $w$. Let $\free(\varphi)$ be the set of all free variables in $\varphi$, and let $\mathcal{V}$ be a finite set of variables containing $\free(\varphi)$. A $(\mathcal{V}, w)$-\emph{assignment} $\sigma$ is a function assigning to every first-order variable of $\mathcal{V}$ an element of $[w]$ and to every second order variable a subset of $[w]$. We define %
$\sigma[x \rightarrow i]$ as the $(\mathcal{V}\cup \{ x\}, w)$-assignment mapping $x$ to $i$ and equaling $\sigma$ everywhere else.
The assignment $\sigma[X \rightarrow I]$ is defined analogously.
\par
Consider the extended alphabet $\Sigma_\mathcal{V}=\Sigma\times\{0,1\}^\mathcal{V}$ together with its natural OPM $M_\mathcal{V}$ defined such that for all $(a,s), (b,t) \in \Sigma_\mathcal{V}$ and all $\odot \in \{\lessdot, \doteq, \gtrdot\}$, we have $(a,s) \odot (b,t)$ if and only if $a\odot b$.
We represent the word $w$ together with the assignment $\sigma$ as a word $(w,\sigma)$ over $(\Sigma_\mathcal{V},M_\mathcal{V})$ such that $1$ denotes every position where $x$ resp. $X$ holds.
A word over $\Sigma_\mathcal{V}$ is called \emph{valid}, if every first-order variable is assigned to exactly one position. Being valid is a regular property which can be checked by an OPA.
\par
We define the \emph{semantics} of $\varphi \in \MSO(\mathbb{K})$
as a function $\llbracket \varphi \rrbracket_\mathcal{V}: (\Sigma_\mathcal{V},M)^+ \rightarrow K$
inductively for all valid $(w,\sigma) \in (\Sigma_\mathcal{V},M)^+$, 
as seen in Fig. \ref{figure:semantics}.
For not valid $(w,\sigma)$, we set $\llbracket \varphi \rrbracket_\mathcal{V}(w,\sigma)=0$.
We write $\llbracket \varphi \rrbracket$ for $\llbracket \varphi \rrbracket_{\free(\varphi)}$.
\par
\setlength{\intextsep}{11pt}
\begin{figure}[ht]
	\begin{tabular}{l l r l }
		$\llbracket \beta \rrbracket_\mathcal{V}(w,\sigma)$ &
		\multicolumn{3}{l}{
			$=\left \{ \begin{array}{ll}
				1 &,\text{ if }(w,\sigma) \models \beta \\ %
				0 &, \text{ otherwise}
					\end{array} \right .$
		}\\
		$\llbracket k \rrbracket_\mathcal{V}(w,\sigma)$ &
		\multicolumn{3}{l}{
			$=~k \quad \text{ for all } k \in \mathbb{K}$
		}\\
		$\llbracket \varphi \oplus \psi \rrbracket_\mathcal{V}(w,\sigma)$ &
		\multicolumn{3}{l}{
			$=\llbracket \varphi \rrbracket_\mathcal{V}(w,\sigma) + \llbracket \psi \rrbracket_\mathcal{V}(w,\sigma)$
		}
		\\
		$\llbracket \varphi \otimes \psi \rrbracket_\mathcal{V}(w,\sigma)$ &
		\multicolumn{3}{l}{
		$=\llbracket \varphi \rrbracket_\mathcal{V}(w,\sigma) \odot \llbracket \psi \rrbracket_\mathcal{V}(w,\sigma)$
		}
		\\
		$\llbracket \bigoplus_x\varphi \rrbracket_\mathcal{V}(w,\sigma)$ &
		\multicolumn{3}{l}{
			$=\sum\limits_{i \in |w|} \llbracket \varphi \rrbracket_{ \mathcal{V} \cup \{x\}}(w, \sigma[x \rightarrow i])$
		}\\
		$\llbracket \textstyle\bigoplus_X\varphi \rrbracket_\mathcal{V}(w,\sigma)$ &
		\multicolumn{3}{l}{
			$=\sum\limits_{I \subseteq |w|} \llbracket \varphi \rrbracket_{ \mathcal{V} \cup \{X\}}(w, \sigma[X \rightarrow I])$
		}\\
		$\llbracket \textstyle\prod_x \varphi \rrbracket_\mathcal{V}(w,\sigma) $&
		\multicolumn{3}{l}{	
			$=\prod\limits_{i \in |w|} \llbracket \varphi \rrbracket_{ \mathcal{V} \cup \{x\}}(w, \sigma[x \rightarrow i])$
		}
	\end{tabular}
	\caption{Semantics} \label{figure:semantics}
\end{figure}
We write $\llbracket \varphi \rrbracket$ for $\llbracket \varphi \rrbracket_{\free(\varphi)}$, so
$\llbracket \varphi \rrbracket : (\Sigma_{\free(\varphi)},M)^+ \rightarrow K$.
If $\varphi$ contains no free variables, $\varphi$ is a \emph{sentence} and
$\llbracket \varphi \rrbracket : (\Sigma,M)^+\rightarrow K$.
\begin{example} \label{Example:exWformula}
Let us go back to the automaton $\mathcal{A}_{\text{policy}}$ depicted in Figure \ref{figure:wOPA2}. The following boolean formula $\beta$ defines three subsets of string positions, $X_0, X_1, X_2$, representing, respectively, the string portions where unmatched calls are not penalized, namely $X_0, X_2$, and the portion where they are, namely $X_1$.
\begin{align*}
	\beta = %
	&\hphantom{{}\land{}} x \in X_0 \leftrightarrow \exists y \exists z ( y>x \land z>x \land \Lab_\$(y) \land \Lab_\$(z)) \\  
	&\land x \in X_1 \leftrightarrow \exists y \exists z
    	\left( 
    	\begin{array}{c}
    		y \le x \le z \land \Lab_\$(y) \land \Lab_\$(z) \\ 
    		\land (x \ne y \land x \ne z \to \neg \Lab_\$(x))
    	\end{array}
    	\right) \\  
	&\land x \in X_2 \leftrightarrow \exists y \exists z ( y<x \land z<x \land \Lab_\$(y) \land \Lab_\$(z)) \enspace.  
\end{align*}
Weight assignment is formalized by
\[\varphi_{0,2} = \neg ((x \in X_0 \lor x \in X_2) \land (\Lab_{\call}(x) \lor \Lab_{\ret}(x) \lor \Lab_{\intt}(x))) \oplus 0 \enspace,
\]  
which assigns weight $0$ to calls, returns, and ints outside portion $X_1$; and 
\begin{align*}
	\varphi_1 =
	&\hphantom{{}\otimes{}} (\neg (x \in X_1 \land \Lab_{\call}(x)) \oplus 1) \\
	&\otimes (\neg (x \in X_1 \land \Lab_{\ret}(x)) \oplus -1) \\
	&\otimes (\neg (x \in X_1 \land \Lab_{\intt}(x)) \oplus 0) \\
	&\otimes (\neg \Lab_\$(x) \oplus 0) \enspace,
\end{align*}
which assigns weights $1, -1, 0$ to calls, returns, and ints, respectively, within portion $X_1$.
\par
Then, the formula $\psi = \prod_x (\beta \otimes \varphi_{0,2} \otimes \varphi_1)$ defines the weight assigned by $\mathcal{A}_{\text{policy}}$ to an input string through a single nondeterministic run and finally $\chi = \bigoplus_{X_0}\bigoplus_{X_1}\bigoplus_{X_2} \psi$ defines the global weight of every string in an equivalent way as the one defined by $\mathcal{A}_{\text{policy}}$. 
\end{example}
\begin{lemma}
	\label{lemma:consistency}
	Let $\varphi \in \MSO(\mathbb{K})$ and let $\mathcal{V}$ be a finite set of variables with $\free(\varphi) \subseteq \mathcal{V}$. Then, $\llbracket \varphi \rrbracket_{\mathcal{V}}(w,\sigma)= \llbracket \varphi \rrbracket(w,\sigma{\restriction_{\free(\varphi)}})$ for each valid $(w,\sigma)\in (\Sigma_\mathcal{V},M)^+$. Furthermore, $\llbracket \varphi \rrbracket$ is regular (resp. strictly regular) iff $\llbracket \varphi \rrbracket_{\mathcal{V}}$ is regular (resp. strictly regular).
\end{lemma}
\begin{proof}
	This is shown by means of %
	Proposition \ref{prop:closurehom} analogously to Proposition 3.3 of \cite{DG07}.
	\qed
\end{proof}

As shown by \cite{DG07} in the case of words, the full weighted logic is strictly more powerful than weighted automata. A similar example also applies here. Therefore, in the following, we restrict our logic in an appropriate way. The main idea for this is to allow only functions with finitely many different values (step functions) after a product quantification. Furthermore, in the non-commutative case, we either also restrict the application of $\otimes$ to step functions or we enforce all occurring weights (constants) of $\varphi \otimes \theta$ to commute.

\begin{definition}
	The \emph{set of almost boolean formulas} is the smallest set of all formulas of $\MSO(\mathbb{K})$ containing all constants $k \in \mathbb{K}$ and all boolean formulas which is closed under $\oplus$ and $\otimes$.
\end{definition}
The following propositions show that almost boolean formulas are describing precisely a certain form of rwOPA's behaviors, which we call \emph{OPL step functions}. %
We adapt ideas from \cite{DM12}.
\begin{definition}
	For $k \in \mathbb{K}$ %
    and a language $L\subseteq (\Sigma,M)^+$, we define $\mathbbm{1}_{L}:(\Sigma,M)^+ \rightarrow \mathbb{K}$, the \emph{characteristic series} of $L$, i.e. $\mathbbm{1}_{L}(w)=1$ if $w \in L$, and $k\mathbbm{1}_{L}(w)=0$ otherwise.
    We denote by $k\mathbbm{1}_{L}:(\Sigma,M)^+ \rightarrow \mathbb{K}$ the characteristic series of $L$ multiplied by $k$, i.e. $k\mathbbm{1}_{L}(w)=k$ if $w \in L$, and $k\mathbbm{1}_{L}(w)=0$ otherwise. 
	\par
	A series $S$ is called an \emph{OPL step function}, if it has a representation %
	\begin{align*}
      S=\sum_{i=1}^n k_i \mathbbm{1}_{L_i}\enspace ,
	\end{align*}
where $L_i$ are OPL
forming a partition of $(\Sigma,M)^+$ and $k_i \in \mathbb{K}$ for each $i \in \{1,...,n\}$; so $\llbracket\varphi\rrbracket(w)=k_i$ iff $w \in L_i$, for each $i \in \{1,...,n\}$.
\end{definition}
\begin{lemma}\label{lemma:stepClosure}
	The set of all OPL step functions is closed under $+$ and $\odot$.
\end{lemma}
\begin{proof}
	Let $S=\sum_{i=1}^k k_i \mathbbm{1}_{L_i}$ and $S'=\sum_{j=1}^\ell k'_j \mathbbm{1}_{L'_j}$ be OPL step functions. 
	Then the following holds
	\begin{align*}
		S+S'
		&=\sum_{i=1}^k \sum_{j=1}^\ell
			(d_i+d'_j)
			\mathbbm{1}_{L_i \cap L'_j}
		\enspace, \\
		S \odot S'
		&=\sum_{i=1}^k \sum_{j=1}^\ell
			(d_i \cdot d'_j)
			\mathbbm{1}_{L_i \cap L'_j} \enspace .
	\end{align*}
	Since $(L_i \cap L'_j)$ are also OPL and form a partition of $(\Sigma,M)^+$, it follows that $S + S'$ and $S  \odot S'$ are also OPL step functions.
\qed
\end{proof}
\begin{proposition} \label{prop:almboolrsf}
\begin{enumerate}[(a)]
	\item For every almost boolean formula $\varphi$, $\llbracket\varphi\rrbracket$ is an OPL step function.
    \item If $S$ is an OPL step function, then there exists an almost boolean formula $\varphi$ such that $ S =\llbracket \varphi \rrbracket$.
\end{enumerate}
\end{proposition}
\begin{proof}
(a) We show the first statement by structural induction on $\varphi$. If $\varphi$ is boolean, then $\llbracket\varphi\rrbracket = \mathbbm{1}_{L(\varphi)}$, were $L(\varphi)$ and $L(\neg\varphi)$ are OPL due to Theorem \ref{thm:OPLMSO}. Therefore,
$\llbracket\varphi\rrbracket = 1_K\mathbbm{1}_{L(\varphi)} + 0_K \mathbbm{1}_{L(\neg\varphi)}$ is an OPL step function. If $\varphi = k$, $k\in \mathbb{K}$, then $\llbracket k \rrbracket = k \mathbbm{1}_{(\Sigma,M)^+}$ is an OPL step function.
Let $\mathcal{V}=\free(\varphi_1) \cup \free(\varphi_2)$. By lifting Lemma \ref{lemma:consistency} to OPL step functions as in \cite{DP14} and by Lemma \ref{lemma:stepClosure}, we see that $\llbracket \varphi_1 \oplus \varphi_2 \rrbracket=\llbracket \varphi_1 \rrbracket_\mathcal{V} + \llbracket \varphi_2 \rrbracket_\mathcal{V}$ and $\llbracket \varphi_1 \otimes \varphi_2 \rrbracket=\llbracket \varphi_1 \rrbracket_\mathcal{V} \odot \llbracket \varphi_2 \rrbracket_\mathcal{V}$ are also OPL step functions.
\par
(b) Given an OPL step function $\llbracket\varphi\rrbracket=\sum_{i=1}^n k_i \mathbbm{1}_{L_i}$, we use Theorem \ref{thm:OPLMSO} to get $\varphi_i$ with $\llbracket\varphi_i \rrbracket= \mathbbm{1}_{L_i}$. Then, the second statement follows from setting $\varphi = \bigvee_i^n(k_i \wedge \varphi_i)$ and the fact that the OPL $(L_i)_{1\leq i \leq n}$ form a partition of $(\Sigma,M)^+$.
\qed
\end{proof}
\begin{proposition}\label{prop:stepTorwOPA}
	Let $S$ be an OPL step function. Then $S$ is strictly regular. 
\end{proposition}
\begin{proof}
	Let $n \in \mathbb{N}$, $(L_i)_{1\leq i \leq n}$ be OPL forming a partition of $(\Sigma,M)^+$ and $k_i \in \mathbb{K}$ for each $i \in \{1,...,n\}$ such that
	\[
	 	S =\sum_{i=1}^n k_i \mathbbm{1}_{L_i} \enspace .
	\]
    Its easy to construct a 2 state rwOPA recognizing the constant series $\llbracket k_i \rrbracket$ which assigns the weight $k_i$ to every word. %
	Hence, $k_i\mathbbm{1}_{L_i} = \llbracket k_i \rrbracket \cap L_i$ is strictly regular by Proposition \ref{prop:restrictedprod}.
	Therefore, %
    by Proposition \ref{prop:disjsums}, $S$ is strictly regular.
	\qed
\end{proof}
\begin{definition}%
	Let $\varphi \in \MSO(\mathbb{K})$.
	We denote by $\const(\varphi)$ all weights of $\mathbb{K}$ occurring in $\varphi$ and we call $\varphi$
		\emph{$\otimes$-restricted} if for all subformulas $\psi \otimes \theta$ of $\varphi$ either
		$\psi$ is almost boolean or $\const(\psi)$ and $\const(\theta)$ commute elementwise. 
	We call $\varphi$
		\emph{$\prod$-restricted} if for all subformulas $\prod_x \psi$ of $\varphi$, %
		$\psi$ is almost boolean.
	We call $\varphi$ \emph{restricted} if it is both $\otimes$- and $\prod$-restricted. %
\end{definition}
In Example \ref{Example:exWformula}, the formula $\beta$ is boolean, the formulas $\phi$ are almost boolean, and $\psi$ and $\chi$ are restricted. Notice that $\psi$ and $\chi$ would be restricted even if $\mathbb{K}$ were not commutative.
\par
For use in Section \ref{section:char}, we note:
\begin{proposition} \label{prop:multiplyk}
Let $S:(\Sigma,M)^+\rightarrow K$ be a regular (resp. strictly regular) series and $k\in\mathbb{K}$. Then $\llbracket k \rrbracket \odot S$ is regular (resp. strictly regular).
\end{proposition}
\begin{proof}
Let $\mathcal{A}=(Q,I,F,\delta,\wt)$ be an (r)wOPA such that $\llbracket \mathcal{A} \rrbracket= S$. Then we construct an rwOPA $\mathcal{B}=(Q',I',F,\delta',\wt')$ as follows.
	\par
    We set $Q\cup I'$ and $I'=\{q'_I \mid q_I \in I\}$.
	The new transition relations $\delta'$ and weight functions $\wt'$ consists of all transitions of $\mathcal{A}$ with their respective weights and the following %
    additional transitions:
    For every push transition $(q_I,a,q)$ of $\delta_\push$, %
    we add %
    a push transition $(q'_I,a,q)$ to $\delta'_\push$ with %
    $\wt'_\push(q_I',a,q)=k\cdot\wt_\push(q_I,a,q)$.
    \par
    Note that every run of an (w)OPA has to start with a push transition.
	Therefore, the two automata have the same respective runs, but $\mathcal{B}$ is exactly once in a state $q'_I \in I$. %
	 This together with the weight assignment ensures that $\mathcal{B}$ uses the same weights as $\mathcal{A}$ except at the very first transition of every run which is multiplied by $k$ from the left. %
    In particular, we do not change the weight of any pop transition.
	It follows that $\llbracket\mathcal{B}\rrbracket=\llbracket k \rrbracket\odot S$. Also, if $\mathcal{A}$ is restricted, so is $\mathcal{B}$.
\qed
\end{proof}
	\section{Characterization of Regular Series} \label{section:char}
\begin{lemma}[Closure under weighted disjunction]
	\label{lemma:closureplus}
	Let $\varphi$ and $\psi$ be two formulas of $\MSO(\mathbb{K})$ such that $\llbracket \varphi \rrbracket$ and $\llbracket \psi \rrbracket$ are regular (resp. strictly regular). Then, $\llbracket \varphi \oplus \psi \rrbracket$ is regular (resp. strictly regular).
\end{lemma}
\begin{proof}
	We put $\mathcal{V}=\free(\varphi) \cup \free(\psi)$. Then, $\llbracket \varphi \oplus \psi \rrbracket= \llbracket \varphi \rrbracket_\mathcal{V}+ \llbracket \psi \rrbracket_\mathcal{V}$ is regular (resp. strictly regular) by Lemma \ref{lemma:consistency} and Proposition \ref{prop:disjsums}.
	\qed
\end{proof}

\begin{proposition}[Closure under restricted weighted conjunction]
	\label{prop:closuredot} %
	Let $\psi \otimes \theta$ be a subformula of a $\otimes$-restricted formula $\varphi$ of $\MSO(\mathbb{K})$ such that $\llbracket \psi \rrbracket$ and $\llbracket \theta \rrbracket$ are regular (resp. strictly regular). Then, $\llbracket \psi \otimes \theta \rrbracket$ is regular (resp. strictly regular).
\end{proposition}
\begin{proof}%
	Since $\varphi$ is $\otimes$-restricted, either $\psi$ is almost boolean or the constants of both formulas commute.
	\par
	\textbf{Case 1:} Let us assume $\psi$ is almost boolean. Then, we can write $\llbracket\psi\rrbracket$ as OPL step function, i.e., $\llbracket\psi\rrbracket =\sum_{i=1}^n k_i \mathbbm{1}_{L_i}$, where $L_i$ are OPL.
    So, %
    the series $\llbracket \psi \otimes \theta \rrbracket$ equals a sum of series of the form $(\llbracket k_i \otimes \theta\rrbracket \cap L_i)$.
    Then, by Proposition \ref{prop:multiplyk}, $\llbracket k_i \otimes \theta\rrbracket$ is a regular (resp. strictly regular) series. 
	Therefore, $(\llbracket k_i \otimes \theta\rrbracket \cap L_i)$ is regular (resp. strictly regular) by Proposition \ref{prop:restrictedprod}. Hence, $\llbracket \psi \otimes \theta \rrbracket$ is (strictly) regular by Proposition \ref{prop:disjsums}.
	\par
	\textbf{Case 2:} Let us assume that the constants of $\psi$ and $\theta$ commute. Then, %
	the second part of Proposition \ref{prop:restrictedprod} yields the claim.
	\qed
\end{proof}

\begin{lemma}[Closure under $\sum_x,~\sum_X$]
	\label{lemma:closureexists}
	Let $\varphi$ be a formula of $\MSO(\mathbb{K})$ such that $\llbracket \varphi \rrbracket$ is regular (resp. strictly regular). Then, $\llbracket \sum_x \varphi \rrbracket$ and $\llbracket \sum_X \varphi \rrbracket$ are regular (resp. strictly regular). 
\end{lemma}
\begin{proof} [Compare \cite{DG07}]	%
	Let $\mathcal{X} \in \{x, X\}$ and $\mathcal{V}=\free(\sum_\mathcal{X} \varphi)$. We define $\pi:(\Sigma_{\mathcal{V} \cup \{\mathcal{X}\} }, M)^+ \rightarrow (\Sigma_\mathcal{V},M)^+$ by $\pi(w,\sigma)=(w,\sigma {\restriction_{\mathcal{V}}})$ for any $(w,\sigma) \in (\Sigma_{\mathcal{V} \cup \{\mathcal{X}\} },M)^+$. Then, for $(w, \gamma) \in (\Sigma_\mathcal{V},M)^+$, the following holds
	\begin{align*}	%
	\llbracket {\textstyle\sum_X} \varphi \rrbracket(w,\gamma) &= \sum_{I \subseteq \{1,...,|w|\}} \llbracket \varphi \rrbracket_{\mathcal{V} \cup \{X\} }(w,\gamma[X \rightarrow I]) \\
	&= \sum_{(w,\sigma)\in \pi^{-1}(w,\gamma) } \llbracket \varphi \rrbracket_{\mathcal{V} \cup \{ X \} } (w,\sigma) \\
	&= \pi(\llbracket \varphi \rrbracket_{\mathcal{V} \cup \{X \} } ) (w, \gamma)\enspace.
	\end{align*}
	Analogously, we show that
	$\llbracket \sum_x \varphi \rrbracket(w,\gamma)=\pi(\llbracket \varphi \rrbracket_{\mathcal{V} \cup \{x \} } ) (w, \gamma)$ for all $(w, \gamma) \in (\Sigma_\mathcal{V},M)^+$.
	By Lemma \ref{lemma:consistency}, $\llbracket \varphi \rrbracket_{\mathcal{V} \cup \{\mathcal{X} \} }$ is regular because $\free(\varphi) \subseteq \mathcal{V} \cup \{\mathcal{X} \}$. Then, $\llbracket \sum_\mathcal{X} \varphi \rrbracket$ is regular by %
	Proposition \ref{prop:closurehom}.
	\qed
\end{proof}
\begin{proposition} [Closure under restricted $\prod_x$]
	\label{prop:closureforall}
	Let $\varphi$ be an almost boolean formula of $\MSO(\mathbb{K})$.
	Then, $\llbracket \prod_x \varphi \rrbracket$ is strictly regular.
\end{proposition}
\begin{proof}
	We use ideas of \cite{DG07} and the extensions in \cite{DPi14} and \cite{DD16} with the following intuition. %
	\par
	In the first part, we write $\llbracket \varphi \rrbracket$ as OPL step function and encode the information to which language $(w,\sigma [x \rightarrow i])$ belongs in a specially extended language $\tilde{L}$. Then we construct an MSO-formula for this language. Therefore, by Theorem \ref{thm:OPLMSO}, we get a deterministic OPA recognizing $\tilde{L}$. %
	In the second part, we add the weights $k_i$ to this automaton and return to our original alphabet.
	\par
	More detailed, let $\varphi \in \MSO(\mathbb{K},(\Sigma,M))$. We define $\mathcal{V}=\free(\prod x. \varphi)$ and $\mathcal{W}=\free(\varphi) \cup \{x\}$. We consider the extended alphabets $\Sigma_\mathcal{V}$ and $\Sigma_\mathcal{W}$ together with their natural OPMs $M_\mathcal{V}$ and $M_\mathcal{W}$.
	By Proposition \ref{prop:almboolrsf} and
    lifting Lemma \ref{lemma:consistency} to OPL step functions, $\llbracket \varphi \rrbracket$ is an OPL step function. Let $\llbracket \varphi \rrbracket = \sum^m_{j=1} k_j \mathbbm{1}_{L_j} $ where $L_j$ is an OPL over $(\Sigma_{\mathcal{W}},M_\mathcal{W})$ for all $j \in \{1,...,m\}$ and $(L_j)$ is a partition of $(\Sigma_{\mathcal{W}},M_\mathcal{W})^+$. 
	By the semantics of the product quantifier, we get
	\begin{align}
	\notag %
	\llbracket {\textstyle\prod_x} \varphi \rrbracket (w,\sigma)
	&= \prod_{i \in [w]}(\llbracket \varphi \rrbracket_{\mathcal{W}} (w,\sigma [x \rightarrow i])) \\
	\notag	&= \prod_{i \in [w]}(k_{g(i)}), \\
	\label{formeld3}
	\text{where} \quad\quad\quad\quad%
	g(i)&=	\left\{\begin{array}{ll}
	1 &, \text{ if } (w,\sigma[x\rightarrow i]) \in L_1 \\
	... &~ \\
	m &, \text{ if } (w,\sigma[x\rightarrow i]) \in L_m \end{array}\right . ,~\text{for all } i \in [w]\enspace. \hspace{0,4cm} %
	\end{align}
	Now, in the first part, we encode the information to which language $(w,\sigma [x \rightarrow i])$ belongs in a specially extended language $\tilde{L}$ and construct %
	an MSO-formula for this language.
	We define the extended alphabet $\tilde{\Sigma}=\Sigma\times \{1,...,n\}$, together with its natural OPM $\tilde{M}$ which only refers to $\Sigma$, so:
	\begin{align*}(\tilde{\Sigma}_\mathcal{V},\tilde{M}_\mathcal{V})^+=\{(w,g,\sigma)~|~(w,\sigma) \in (\Sigma_\mathcal{V},M_V) \text{ and } g \in \{1,...,m \}^{[w]} \}\enspace. \end{align*} %
	We define the languages $\tilde{L}, \tilde{L}_j, \tilde{L}'_j \subseteq (\tilde{\Sigma}_\mathcal{V},\tilde{M}_\mathcal{V})^+$ as follows:
	\begin{align*}
	\notag \tilde{L}=	&
	\Set{(w,g,\sigma)|
		\begin{aligned}
		&(w,\sigma) \in (\tilde{\Sigma}_\mathcal{V},\tilde{M}_\mathcal{V})^+ \text{ is valid and }\\
		&\text{for all } i \in [w],~j \in \{1,...,m \}: %
		~g(i)=j \Rightarrow (w,\sigma [x \rightarrow i]) \in L_j
		\end{aligned}
	}
	\enspace, \\
	\notag \tilde{L}_j=	&
	\Set{(w,g,\sigma)|
		\begin{aligned}
		&(w,\sigma) \in (\tilde{\Sigma}_\mathcal{V},\tilde{M}_\mathcal{V})^+ \text{ is valid and} \\ 
		& \text{for all } i \in [w]: %
		~g(i)=j \Rightarrow (w,\sigma [x \rightarrow i]) \in L_j
		\end{aligned}
	}	
	\enspace, \\
	\tilde{L}_j'=	&
	\Set{(w,g,\sigma)|
		\text{for all } i \in [w]: %
		~g(i)=j \Rightarrow (w,\sigma [x \rightarrow i]) \in L_j
	}
	\enspace.
	\end{align*}
	Then, $ \tilde{L}=\bigcap^m_{j=1} \tilde{L}_j$. Hence, in order to show that $\tilde L$ is an OPL, it suffices to show that each $\tilde L_j$ is an OPL. %
	By a standard procedure, compare \cite{DG07}, we obtain a formula $\tilde\varphi_j \in \MSO(\tilde\Sigma_\mathcal{V},\tilde{M}_\mathcal{V})$ with $L(\tilde\varphi_j)=\tilde L'_j$.
	Therefore, by Theorem \ref{thm:OPLMSO}, $\tilde{L}_j'$ is an OPL. It is straightforward to define an OPA accepting $\tilde{N}_\mathcal{V}$, the language of all valid words. %
	By closure under intersection, %
	$\tilde{L}_j=\tilde{L}_j' \cap \tilde{N}_\mathcal{V}$ is also an OPL and so is $\tilde L$. Hence, there exists a deterministic OPA $\mathcal{\tilde{A}}=(\Sigma,q_0,F,\tilde\delta)$ recognizing $\tilde{L}$.
	\par
	In the second part, we add weights to $\mathcal{\tilde{A}}$ as follows.
	We construct the wOPA $\mathcal{A}=(Q,I,F,\delta,\wt)$ over $(\Sigma_\mathcal{V},M_\mathcal{V})$ and $\mathbb{K}$ by adding to every transition of $\mathcal{\tilde A}$ with $g(i)=j$ the weight $k_j$. %
	\par
	That is, we keep the states, the initial state, and the accepting states, and for $\delta=(\delta_{\push},\delta_{\shift},\delta_{\pop})$ and all $q,q',p \in Q$ and $(a,j,s) \in \tilde \Sigma_\mathcal{V}$, we define
	\begin{align*}
	\delta_{\push/\shift}(q,(a,s),q')&=\left\{\begin{array}{ll}	k_j &,\text{ if } (q,(a,j,s),q') \in \tilde\delta_{\push/\shift} \\	
	0 &,\text{ otherwise} \end{array}\right . \enspace .
	\end{align*} %
	Since $\mathcal{\tilde A}$ is deterministic, for every $(w,g,\sigma) \in \tilde L$, there exists exactly one accepted run $\tilde r$ of $\mathcal{\tilde A}$. On the other hand, for every $(w,g,\sigma) \notin \tilde L$, there is no accepted run of $\mathcal{\tilde A}$. Since $(L_j)$ is a partition of $(\Sigma_{\mathcal{W}},M_\mathcal{W})^+$, for every $(w,\sigma) \in ( \Sigma_\mathcal{V},M_\mathcal{V})$, there exists exactly one $g$ with $(w,g,\sigma) \in \tilde L$. %
	Thus, every $(w,\sigma) \in (\Sigma_\mathcal{V},M_\mathcal{V})$ has exactly one run $r$ of $\mathcal{A}$ %
	determined by the run $\tilde r$ of $(w,g,\sigma)$ of $\mathcal{\tilde A}$. %
	We denote with $\wt_\mathcal{A}(r,(w,\sigma),i)$ the weight used by the run $r$ on $(w,\sigma)$ over $\mathcal{A}$ at position $i$, which is always the weight of the push or shift transition used at this position. Then by definition of $\mathcal{A}$ and $\tilde L$, the following holds for all $i \in [w]$
	\begin{align*}
	g(i)=j \Rightarrow \wt_\mathcal{A}(r,(w,\sigma),i)=k_j \wedge (w,\sigma[x \rightarrow i]) \in L_j\enspace.
	\end{align*}
	By formula (\ref{formeld3}), we obtain 
	\begin{align*}
	\llbracket \varphi \rrbracket_\mathcal{W} (w,\sigma [x \rightarrow i])=k_j=\wt_\mathcal{A}(r,(w,\sigma),i)\enspace.
	\end{align*}
	Hence, for the behavior of the automaton $\mathcal{A}$ the following holds
	\begin{align*}
	\notag	\llbracket \mathcal{A} \rrbracket (w,\sigma) &=	\sum_{r' \in \acc(\mathcal{A},w)} \wt(r')\\ %
	\notag		&=\prod^{|w|}_{i=1}\wt_\mathcal{A}(r,(w,\sigma),i) \\
	\notag		&=\prod^{|w|}_{i=1}\llbracket \varphi \rrbracket_\mathcal{W} (w,\sigma [x \rightarrow i]) \\
	&=\llbracket \textstyle \prod_x \varphi \rrbracket (w,\sigma)\enspace. %
	\end{align*}
	Thus, $\mathcal{A}$ %
	recognizes $\llbracket \prod_x \varphi \rrbracket$.
	\qed
\end{proof}
The following proposition is a summary of the previous results.
\begin{proposition} \label{prop:FormulaToRec}
	For every restricted
	$\MSO(\mathbb{K})$-sentence $\varphi$, 
	there exists an rwOPA $\mathcal{A}$
	with $\llbracket \mathcal{A} \rrbracket = \llbracket \varphi \rrbracket$.
\end{proposition}
\begin{proof}
We use structural induction on $\varphi$.
    If $\varphi$ is an almost boolean formula, then by Proposition \ref{prop:almboolrsf} $\llbracket\varphi\rrbracket$ is an OPL step function. By Proposition \ref{prop:stepTorwOPA} every OPL step function is strictly regular. %
	\par
	Closure under $\oplus$ is dealt with by Lemma \ref{lemma:closureplus}, closure under $\otimes$ by Proposition \ref{prop:closuredot}. %
	The sum quantifications $\sum_x$ and $\sum_X$ are dealt with by Lemma \ref{lemma:closureexists}.
	Since $\varphi$ is restricted, we know that for every subformula $\bigotimes_x \psi$, the formula $\psi$ is an almost boolean formula. Therefore, we can apply Proposition \ref{prop:closureforall} to maintain recognizability of our formula in this case. %
\end{proof}
The next proposition shows that the converse also holds.
\begin{proposition} \label{prop:RecToFormula}
	For every rwOPA $\mathcal{A}$, there exists a restricted
	$\MSO(\mathbb{K})$-sentence $\varphi$ with $\llbracket \mathcal{A} \rrbracket = \llbracket \varphi \rrbracket$.
	If $\mathbb{K}$ is commutative, then for every wOPA $\mathcal{A}$, there exists a restricted $\MSO(\mathbb{K})$-sentence $\varphi$
	with $\llbracket \mathcal{A} \rrbracket = \llbracket \varphi \rrbracket$.
\end{proposition}

\begin{proof}
	The rationale adopted to build formula $\varphi$ from $\mathcal{A}$ integrates the approach followed in \cite{DG07,DPi14} with the one of \cite{LMPP15}
	On the one hand we need second order variables suitable to ``carry'' weights; on the other hand, unlike previous non-OP cases which are managed through real-time automata, an OPA can perform several transitions while remaining in the same position. Thus, we introduce the following second order variables: 
	$X_{p,a,q}^{\push}$ represents the set of positions where $\mathcal{A}$ performs a push move from state $p$, reading symbol $a$ and reaching state $q$; $X_{p,a,q}^{\shift}$ has the same meaning as $X_{p,a,q}^{\push}$ for a shift operation; $X_{p,q,r}^{\pop}$ represents the set of positions of the symbol that is on top of the stack when $\mathcal{A}$ performs a pop transition from state $p$, with $q$ on top of the stack, reaching $r$. 

\begin{figure} 
\begin{adjustbox}{max totalsize={\textwidth},center}
		\begin{tikzpicture}[flush/.style={double, >=stealth, thin, rounded corners}]
		\matrix (m) [matrix of nodes]
		{
			& & $$ & &  & & & & & &  & & & &  & & & & & & \\
			$\circ$ & & & & $X^{\pop}_{3,1,3}$ & & & & & & & & & & & & & & & & \\
			
			& & & &  & & $$ & &  & & & & & & & & & & & & \\
			& & & & $\circ$ & & & & $X^{\pop}_{3,1,3}$ & & & & & & & & & & & & \\
			
			& & & & & & & &  & & & & & $$ & & & & &  & & \\
			& & & & & & & & $\circ$ & & & & & & & & & & $X^{\pop}_{3,0,3}$ & & \\
			
			& & & & & & & & & &  & & $$ & &  & & & & & & \\
			& & & & & & & & & & $\circ$ & & & & $X^{\pop}_{3,3,3}$ & & & & & & \\
			& $$ &  & &  & $$ &  & & & &  & $$ &  & &  & $$ &  & & & & & & \\
			
			$\circ$ &  & $X^{\pop}_{1,0,1}$ & & $\circ$ &  & $X^{\pop}_{1,0,1}$ & & & & $\circ$ &  & $X^{\pop}_{3,2,3}$ & & $\circ$ &  & $X^{\pop}_{3,2,3}$ & & & & & & \\
			
			& & $X^{\push}_{0,n,1}$ & & $X^{\push}_{1,+,0}$ & & $X^{\push}_{0,n,1}$ & & $X^{\push}_{1,\times,0}$ & & $X^{\push}_{0,(,2}$ & & $X^{\push}_{2,n,3}$ & & $X^{\push}_{3,+,2}$ & & $X^{\push}_{2,n,3}$ & & $X^{\shift}_{3,),3}$ & & %
            \\
			
			\# & & $n$ & $\ \ $ & $+$ & & $n$ & $\ \ $ & $\times$ & $\ \ $ & $($ & & $n$ & $\ \ $ & $+$ & & $n$ & $\ \ $ & $)$ & $\ \ $ & \# \\
			0 & & 1 &  & 2  & & 3 &  & 4 &  & 5  &  & 6  &  & 7 &  & 8 &  & 9  &   & 10 \\
		};
		
		\draw[-] (m-10-1)   to [out=65, in=145]  (m-10-3);
		\draw[-] (m-10-5)   to [out=65, in=145]  (m-10-7);
		\draw[-] (m-10-11)  to [out=65, in=145] (m-10-13);
		\draw[-] (m-10-15)  to [out=65, in=145] (m-10-17);
		
		\draw[-] (m-2-1)  to [out=50, in=145]  (m-2-5); 
		\draw[-] (m-8-11) to [out=50, in=145]  (m-8-15);
		
		\draw[-] (m-4-5)  to [out=40, in=145] (m-4-9); 
		
		\draw[-] (m-6-9)  to [out=30, in=155] (m-6-19);
		\end{tikzpicture}
        \end{adjustbox}
		\caption{The string of Figure~\ref{fig:exp} with the second order variables evidenced for the automaton of Figure~\ref{fig:exprAut}. The symbol $\circ$  marks the positions of the symbols that precede the push corresponding to the bound pop transition.}
    \label{fig:log}
	\end{figure}
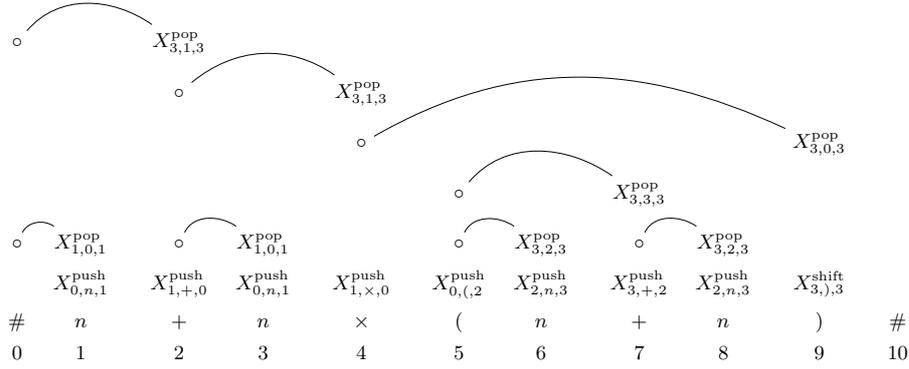
	\par
	Let $\mathcal{V}$ consist of all $X_{p,a,q}^{\push}$, $X_{p,a,q}^{\shift}$, and $X_{p,q,r}^{\pop}$ such that $a \in \Sigma$, $p,q,r \in Q$ and $(p,a,q) \in \delta_\push$ resp. $\delta_\shift$, resp. $(p,q,r) \in \delta_\pop$.
    Since $\Sigma$ and $Q$ are finite,
    there is an enumeration $\bar{X}=(X_1,..,X_m)$ of all variables of $\mathcal{V}$. %
	We denote by $\bar{X}^{\push}$, $\bar{X}^{\shift}$, and  $\bar{X}^{\pop}$ enumerations over only the respective set of second order variables.
	\par
	We use the following usual abbreviations for unweighted formulas of $\MSO$:
	\begin{align*}
		(\beta \wedge \varphi) &= \neg(\neg \beta \vee \neg \varphi),\\
		(\beta \rightarrow \varphi) &= (\neg \beta \vee \varphi),\\
		(\beta \leftrightarrow \varphi) &= (\beta \rightarrow \varphi) \wedge (\varphi \rightarrow \beta),\\
		(\forall x.\varphi )&= \neg (\exists x. \neg \varphi),\\
		(y=x) &= (x \le y) \wedge (y \le x), \\
		(y=x+1) &=(x \leq y) \wedge \neg(y \leq x) \wedge \forall z.(z \leq x \vee y\leq z), \\
		\min(x) &= \forall y. ( x \leq y), \\
		\max(x) &= \forall y. ( y \leq x), \\
	\end{align*}
	Additionally, we use the shortcuts $\Tree(x,z,v,y)$, $\Next_i(x,y)$, $Q_i(x,y)$, and $\Tree_{p,q}(x,z,v,y)$, originally defined in \cite{LMPP15}, reported and adapted here for convenience:
    \[
		x \circ y  :=  
        	\bigvee_{a,b \in \Sigma, M_{a,b} = \circ} \Lab_a(x) \land \Lab_b(y), 
			\text{ for }  
			\circ \in \{\lessdot, \doteq, \gtrdot \}
            \]
\[
		\Tree (x,z,v,y)  :=  
			x \avv y \land
			\left(
				\begin{array}{c}
					( x+1 = z \ \lor\  x \avv z )
                    	\land
                    	\neg \exists t ( z < t < y \land x \avv t ) 
					\\ \land \\
					( v+1 = y \ \lor\ v \avv y )
                    	\land
                    	\neg \exists t ( x < t < v \land t \avv y ) 
				\end{array}
			\right)
          \]
      In other words, Tree holds among the four positions $(x, z, v, y)$ iff, at the time when a pop transition is executed:
$x$ (resp. $y$) is the rightmost leaf at the left (resp. the leftmost at the right) of the subtree whose scanning (and construction if used as a parser) is completed by the OPA through the current transition;
$z$ and $y$ are the leftmost and rightmost terminal characters of the right hand side of the grammar production that is reduced by the pop transition of the OPA~\cite{LMPP15}.
For instance, with reference to Figures~\ref{fig:exp} and \ref{fig:log}, $\Tree(5,7,7,9)$ and $\Tree(4,5,9,10)$ hold.
          \[
		\Succ_q (x, y)  :=  
        	(x+1 = y) \land
        	\bigvee_{p \in Q, a \in \Sigma}
            (
        		x \in X^{\push}_{p,a,q}
            	\lor x \in X^{\shift}_{p,a,q}
                \lor \min(x)	%
            )
            \]
 I.e., $y$ is the position adjacent to $x$, $\Lab_a(y)$ and, while reading $a$, the OPA reaches state $q$, either through a push or through a shift move.         

\[
        \Next_r (x,y)  := 
		\exists z \exists v .  
		\left(\Tree(x,z,v,y) \land 
		\bigvee_{p,q \in Q}
		v \in X^{\pop}_{p,q,r}\right)
        \]
      I.e., $\Next_r (x,y)$ holds when a pop move reduces a subtree enclosed between positions $x$ and $y$ reaching state $r$.
		\[
		Q_i(x,y)  := 
        	\Succ_i (x,y) \lor \Next_i (x,y)
        \]
        Finally, 
        \[
        \Tree_{i,j} (x,z,v,y)  :=  
        	\Tree(x,z,v,y) \land Q_i(v,y) \land Q_j(x,z) %
            \]
refines the predicate $\Tree$ by making explicit that $i$ and $j$ are, respectively, the current state and the state on top of the stack when the pop move is executed.

	We now define the unweighted formula $\psi$ to characterize all accepted runs
	\begin{align*}\psi&=Partition(\bar{X}^{\push},\bar{X}^{\shift})\wedge Unique(\bar{X}^{\pop}) %
	\wedge InitFinal \\
	&~~\wedge
	Trans_{\push} \wedge Trans_{\shift} \wedge Trans_{\pop} \enspace.\end{align*}
	Here, the subformula $Partition$ will enforce the push and shift sets to be (together) a partition of all positions. 
    $InitFinal$ controls the initial and the acceptance condition and $Trans_{op}$ the transitions of the run together with the labels.
	\begin{align*}
	Parti&tion(X_1,...,X_n)=\forall x. \bigvee^n_{i=1}\big[(x \in X_i) \wedge \bigwedge_{i \neq j} \neg (x \in X_j)\big] \enspace, \\
	Un&ique(X^{\pop}_1,..,X^{\pop}_n)=\forall x. \bigwedge_{i \neq j} \neg (x \in X_i^\pop \wedge x \in X_j^\pop) \enspace, \\
	InitFinal &=\exists x \exists y \exists x' \exists y'. %
		\big [
			\min(x)
			\wedge
			\max(y)
            \wedge
            x + 1 = x'
            \wedge
            y' + 1 = y
		\\ &\hspace{2cm}	\wedge
			\bigvee_{\substack{i \in I,\, q \in Q \\ a\in \Sigma}} %
				x' \in X^{\push}_{i,a,q} %
		\\ &\hspace{2cm}	\wedge
			\bigvee_{\substack{f \in F,\, q \in Q \\ a\in \Sigma}}
				(y' \in X^{\push}_{q,a,f} %
				\vee
				y' \in X^{\shift}_{q,a,f}) %
		\\ &\hspace{2cm}	%
        	\wedge			%
			\bigvee_{f\in F}
            	(\Next_f(x,y)
				\wedge
				\bigwedge_{j\neq f}\neg\Next_j(x,y))
		\big ] \enspace,
    \end{align*}
    \begin{align*}
	Trans_{\push} &= \forall x. \bigwedge_{p,q \in Q, a\in \Sigma}
    \big(
		x\in X^{\push}_{p,a,q} \rightarrow
        \big[
			\Lab_a(x)		%
			\wedge
			\exists z.
            (
				z \lessdot x
				\wedge
                Q_p(z,x)	%
			)
		\big]
    \big)
	\\
	Trans_{\shift} &= \forall x. \bigwedge_{p,q \in Q, a\in \Sigma}
    \big(
		x\in X^{\shift}_{p,a,q} \rightarrow
        \big[
			\Lab_a(x)		%
			\wedge
			\exists z.
            (
				z \doteq x
				\wedge
                Q_p(z,x)	%
			)
		\big]
    \big)
		\enspace.
	\end{align*}
    I.e., if $x \in X^{\push}_{p,a,q}$ (resp. $X^{\shift}$) %
    the formula holds in a run where,
    reading character $a$ in position $x$, the automaton performs a push (resp. a shift) reaching state $q$ from $p$; 
    this may occur when $z \lessdot x$ (resp., $z \doteq x$) is immediately adjacent to $x$ or after a subtree between positions $z$ and $x$ has been built. 
Notice that the converse too of the above implications holds, due to the fact that the whole set of string positions is partitioned into the two disjoint sets $X^\push$, $X^\shift$.
\[
   	Trans_{\pop} = \forall v. \bigwedge_{p,q \in Q}
   	\big(
	   	\big[
		   	\bigvee_{r\in Q} v\in X^{\pop}_{p,q,r}
		\big]
		\leftrightarrow
	   	\big[
		   	\exists x\exists y\exists z.
		   	(
			   	\Tree_{p,q}(x,z,v,y)
			)
	   	\big]
   	\big)
\]

Thus, with arguments similar to \cite{LMPP15} it can be shown that the sentences satisfying $\psi$ are exactly those recognized by the unweighted OPA subjacent to $\mathcal{A}$. %
\par
For an unweighted formula $\beta$ and two weights $k_1$ and $k_2$, we define the following shortcut for an almost boolean weighted formula:
\[
	\myIf \beta \myThen k_1 \myElse k_2 = (\beta \otimes k_1) \oplus (\neg \beta \otimes k_2) \enspace .
\] %
Now, we add weights to $\psi$ by defining the following restricted weighted formula
	\begin{align*}
		\theta=\psi\otimes
		{\textstyle\prod_x}
		\mathop{\otimes}\limits_{p,q \in Q}
		\big (
			&\mathop{\otimes}\limits_{a \in \Sigma}
                (\myIf x\in X^{\push}_{p,a,q} \myThen\wt_{\push}(p,a,q)\myElse 1)\\
			\otimes &\mathop{\otimes}\limits_{a \in \Sigma}
                (\myIf x\in X^{\shift}_{p,a,q} \myThen\wt_{\shift}(p,a,q)\myElse 1) \\
			\otimes &\mathop{\otimes}\limits_{r \in Q}
                (\myIf x\in X^{\pop}_{p,q,r} \myThen\wt_{\pop}(p,q,r)\myElse 1)
		\big )
		\enspace. 
	\end{align*}
    Here, the second part of $\theta$ multiplies up all weights of the encountered transitions. 
	This is the crucial part where we either need that $\mathbb{K}$ is commutative %
    or all pop weights are trivial %
    because the product quantifier of $\theta$ assigns the pop weight at a different position than the occurrence of the respective pop transition in the automaton. Using only one product quantifier (weighted universal quantifier) this is unavoidable, since the number of pops at a given position is only bounded by the word length.
	\par
	Since the subformulas $x\in X ^{()}_{()} \otimes \wt(...)$ of $\theta$ are almost boolean, the subformula $\prod_x (...)$ of $\theta$ is $\prod$-restricted. Furthermore, $\psi$ is boolean and so $\theta$ is $\otimes$-restricted. Thus, $\theta$ is a restricted formula. \par
	Finally, we define
	\begin{align*}\varphi = \textstyle\bigoplus_{X_1}\bigoplus_{X_2}...\bigoplus_{X_m} \theta \enspace.\end{align*}
	This implies $\llbracket \varphi \rrbracket (w) = \llbracket \mathcal{A} \rrbracket (w)$, for all $w \in (\Sigma,M)^+$.
	Therefore, $\varphi$ is our required sentence with $\llbracket \mathcal{A} \rrbracket = \llbracket \varphi \rrbracket$.
	\qed
\end{proof}

The following theorem summarizes the main results of this section.
\begin{theorem}
	\label{main}
	Let $\mathbb{K}$ be a %
	semiring and $S: (\Sigma,M)^+\rightarrow K$ a series.
	\begin{enumerate}%
		\item 	The following are equivalent:
		\begin{enumerate}[(i)]
			\item $S=\llbracket \mathcal{A} \rrbracket$ for some restricted wOPA.
			\item $S=\llbracket \varphi \rrbracket$ for some restricted sentence $\varphi$ of $\MSO(\mathbb{K})$.
		\end{enumerate}
		\item 	Let $\mathbb{K}$ be commutative. Then, the following are equivalent:
		\begin{enumerate}[(i)]
			\item $S=\llbracket \mathcal{A} \rrbracket$ for some wOPA.
			\item $S=\llbracket \varphi \rrbracket$ for some restricted sentence $\varphi$ of $\MSO(\mathbb{K})$.
		\end{enumerate}
	\end{enumerate}
\end{theorem}
Theorem \ref{main} documents a further step in the path of generalizing a series of results beyond the barrier of regular and structured --or visible-- CFLs. Up to a few years ago, major properties of regular languages, such as closure w.r.t. all main language operations, 
decidability results, logic characterization, and, in this case, weighted language versions, could be extended to several classes of structured CFLs, among which the VPL one certainly obtained much attention. OPLs further generalize the above results not only in terms of strict inclusion, but mainly because they are not visible, in the sense explained in the introduction, nor are they necessarily real-time: this allows them to cover important applications that could not be adequately modeled through more restricted classes.

Theorem \ref{main} also shows that the typical logical characterization of weighted languages does not generalize in the same way to the whole class wOPL: for non-rwOPL we need the extra hypothesis that $\mathbb{K}$ be commutative. This is due to the fact that pop transitions are applied in the reverse order than that of positions to which they refer (position $v$ in formula $Trans_{\pop}$). 
Notice, however, that rwOPL do not forbid unbounded pop sequences; thus, they too include languages that are neither real-time nor visible. This remark naturally raises new intriguing questions which we will briefly address in the conclusion.
	\section{Conclusion}\label{section:concl}
    We introduced and investigated weighted operator precedence automata and a corresponding weighted MSO logic.
    In our main results we show, for any semiring, that wOPA without pop weights and a restricted weighted MSO logic have the same expressive power; furthermore, these behaviors can also be described as homomorphic images of the behaviors of particularly simple wOPA reduced to arbitrary unweighted OPA.
    If the semiring is commutative, these results apply also to wOPA with arbitrary pop weights.

This raises the problems to find, for arbitrary semirings and for wOPA with pop weights, both an expressively equivalent weighted MSO logic and a Nivat-type result. In \cite{DV06}, very similar problems arose for weighted automata on unranked trees and weighted MSO logic.
	In \cite{DHV15}, the authors showed that with another definition of the behavior of weighted unranked tree automata, an equivalence result for the restricted weighted MSO logic could be derived.
Is there another definition of the behavior of wOPA (with pop weights) making them expressively equivalent to our restricted weighted MSO logic?

In \cite{LMPP15}, operator precedence languages of infinite words were investigated and shown to be practically important. Therefore, the problem arises to develop a theory of wOPA on infinite words.
	In order to define their infinitary quantitative behaviors, one could try to use valuation monoids as in \cite{DM12}.

Finally, a new investigation field can be opened by exploiting the natural suitability of OPL towards parallel elaboration \cite{BarenghiEtAl2015}. Computing weights, in fact, can be seen as a special case of semantic elaboration which can be performed hand-in-hand with parsing. In this case too, we can expect different challenges depending on whether the weight semiring is commutative or not and/or weights are attached to pop transitions too, which would be the natural way to follow the traditional semantic evaluation through synthesized attributes \cite{knuth1968}.
\bibliography{DDMP17bib}{}

\end{document}